\documentclass[preprint]{sigplanconf}

\usepackage{amsthm}
\usepackage{amssymb}
\usepackage{amsmath}
\usepackage{verbatim}
\usepackage{alltt}
\usepackage{lineno}
\usepackage{trivfloat}
\usepackage{tikz}
\usepackage{pgfplots}
\usepackage[utf8]{inputenc}
\usepackage[T1]{fontenc}

\usetikzlibrary{patterns}
\usetikzlibrary{arrows}

\newfloat{plate}{tbp}{lop}

\theoremstyle{definition}
\newtheorem{definition}{Definition}

\theoremstyle{plain}
\newtheorem{lemma}{Lemma}

\theoremstyle{remark}
\newtheorem{corollary}{Corollary}

\theoremstyle{plain}
\newtheorem{theorem}{Theorem}

\makeatletter
\newcommand{\customlabel}[2]{%
\protected@write \@auxout {}{\string \newlabel {#1}{{#2}{}}}}
\makeatother

\begin{document}

\title{Cache-Aware Lock-Free Concurrent Hash Tries}

\authorinfo{Aleksandar Prokopec, Phil Bagwell, Martin Odersky}
          {École Polytechnique Fédérale de Lausanne, Lausanne, Switzerland}
          {\{firstname\}.\{lastname\}@epfl.ch}

\maketitle

\begin{abstract}
This report describes an implementation of a non-blocking concurrent shared-memory hash trie based on single-word compare-and-swap instructions. Insert, lookup and remove operations modifying different parts of the hash trie can be run independent of each other and do not contend. Remove operations ensure that the unneeded memory is freed and that the trie is kept compact. A pseudocode for these operations is presented and a proof of correctness is given -- we show that the implementation is linearizable and lock-free. Finally, benchmarks are presented which compare concurrent hash trie operations against the corresponding operations on other concurrent data structures, showing their performance and scalability.
\end{abstract}

\section{Introduction}

Many applications access data concurrently in the presence of multiple processors. Without proper synchronization concurrent access to data may result in errors in the user program. A traditional approach to synchronization is to use mutual exclusion locks. However, locks induce a performance degradation if a thread holding a lock gets delayed (e.g. by being preempted by the operating system). All other threads competing for the lock are prevented from making progress until the lock is released. More fundamentally, mutual exclusion locks are not fault tolerant -- a failure may prevent progress indefinitely.

A lock-free concurrent object guarantees that if several threads attempt to perform an operation on the object, then at least some thread will complete the operation after a finite number of steps. Lock-free data structures are in general more robust than their lock-based counterparts \cite{concdatastructures}, as they are immune to deadlocks, and unaffected by thread delays and failures. Universal methodologies for constructing lock-free data structures exist \cite{artofmulti08}, but they serve as a theoretical foundation and are in general too inefficient to be practical -- developing efficient lock-free data structures still seems to necessitate a manual approach.

Trie is a data structure with a wide range of applications first developed by Brandais \cite{brandaistries} and Fredkin \cite{fredkintries}.
Hash array mapped tries described by Bagwell \cite{idealhash} are a specific type of tries used to store key-value pairs. The search for the key is guided by the bits in the hashcode value of the key. Each hash trie node stores references to subtries inside an array which is indexed with a bitmap. This makes hash array mapped tries both space-efficient and cache-aware. A similar approach was taken in the dynamic array data structures \cite{judyarray}.
In this paper we present and describe in detail a non-blocking implementation of the hash array mapped trie data structure.

Our contributions are the following:

\begin{enumerate}
\item We introduce a completely lock-free concurrent hash trie data structure for a shared-memory system based on single-word compare-and-swap instructions. A complete pseudocode is included in the paper.
\item Our implementation maintains the space-efficiency of sequential hash tries. Additionally, remove operations check to see if the concurrent hash trie can be contracted after a key has been removed, thus saving space and ensuring that the depth of the trie is optimal.
\item There is no stop-the-world dynamic resizing phase during which no operation can be completed -- the data structure grows with each subsequent insertion and removal. This makes our data structure suitable for real-time applications.
\item We present a proof of correctness and show that all operations are linearizable and lock-free.
\item We present benchmarks that compare performance of concurrent hash tries against other concurrent data structures. We interpret and explain the results.
\end{enumerate}

The rest of the paper is organized as follows. Section 2 describes sequential hash tries and several attempts to make their operations concurrent. It then presents case studies with concurrent hash trie operations. Section 3 presents the algorithm for concurrent hash trie operations and describes it in detail. Section 4 presents the outline of the correctness proof -- a complete proof is given in the appendix. Section 5 contains experimental results and their interpretation. Section 6 presents related work and section 7 concludes.

\section{Discussion}

Hash array mapped tries (from now on hash tries) described previously by Bagwell \cite{idealhash} are trees which have 2 types of nodes -- internal nodes and leaves. Leaves store key-value bindings. Internal nodes have a $2^W$-way branching factor. In a straightforward implementation, each internal node is a $2^W$-element array. Finding a key proceeds in the following manner. If the internal node is at the root, the initial $W$ bits of the key hashcode are used as an index in the array. If the internal node is at the level $l$, then $W$ bits of the hashcode starting from the position $W * l$ are used. This is repeated until a leaf or an empty entry is found. Insertion and removal are similar.

Such an implementation is space-inefficient -- most entries in the internal nodes are never used. To ensure space efficiency, each internal node contains a bitmap of length $2^W$. If a bit is set, then its corresponding array entry contains an element. The corresponding array index for a bit on position $i$ in the bitmap $bmp$ is calculated as $\#((i - 1) \odot bmp)$, where $\#( \cdot )$ is the bitcount and $\odot$ is a bitwise AND operation. The $W$ bits of the hashcode relevant at some level $l$ are used to compute the index $i$ as before. At all times an invariant is preserved that the bitmap bitcount is equal to the array length. Typically, $W$ is $5$ since that ensures that $32$-bit integers can be used as bitmaps. An example hash trie is shown in Fig. \ref{f-tries}A.

We want to preserve the nice properties of hash tries -- space-efficiency, cache-awareness and the expected depth of $O(\log_{2^W}(n))$, where $n$ is the number of elements stored in the trie and $2^W$ is the bitmap length. We also want to make hash tries a concurrent data structure which can be accessed by multiple threads. In doing so, we avoid locks and rely solely on CAS instructions. Furthermore, we ensure that the new data structure has the lock-freedom property. We call this new data structure a \textit{Ctrie}. In the remainder of this chapter we give examples of Ctrie operations.

Assume that we have a hash trie from Fig. \ref{f-tries}A and that a thread $T_1$ decides to insert a new key below the node \verb=C1=. One way to do this is to do a CAS on the bitmap in \verb=C1= to set the bit which corresponds to the new entry in the array, and then CAS the entry in the array to point to the new key. This requires all the arrays to have additional empty entries, leading to inefficiencies. A possible solution is to keep a pointer to the array inside \verb=C1= and do a CAS on that pointer with the updated copy of the array. The fundamental problem that still remains is that such an insertion does not happen atomically. It is possible that some other thread $T_2$ also tries to insert below \verb=C1= after its bitmap is updated, but before the array pointer is updated. Lock-freedom is not ensured if $T_2$ were to wait for $T_1$ to complete.

Another solution is for $T_1$ to create an updated version of \verb=C1= called \verb=C1'= with the updated bitmap and the new key entry in the array, and then do a CAS in the entry within the \verb=C2= array which points to \verb=C1=. The change is then done atomically. However, this approach does not work. Assume that another thread $T_2$ decides to insert a key below the node \verb=C2= at the time when $T_1$ is creating \verb=C1'=. To do this, it has to read \verb=C2= and create its updated copy \verb=C2'=. Assume that after that, $T_1$ does the CAS in \verb=C2=. The copy \verb=C2'= will not reflect the changes by $T_1$. Once $T_2$ does a CAS in the \verb=C3= array, the key inserted by $T_1$ is lost.

To solve this problem we define a new type of a node which we call an \textit{indirection node}. This node remains present within the Ctrie even if nodes above and below it change. We now show an example of a sequence of Ctrie operations.

Every Ctrie is defined by the \verb=root= reference (Fig. \ref{f-tries}B). Initially, the \verb=root= is set to a special value called \verb=null=. In this state the Ctrie corresponds to an empty set, so all lookups fail to find a value for any given key and all remove operations fail to remove a binding.

Assume that a key $k_1$ has to be inserted. First, a new node \verb=C1= of type \verb=CNode= is created, so that it contains a single key \verb=k1= according to hash trie invariants. After that, a new node \verb=I1= of type \verb=INode= is created. The node \verb=I1= has a single field $main$ (Fig. \ref{f-types}) which is initialized to \verb=C1=. A CAS instruction is then performed at the \verb=root= reference (Fig. \ref{f-tries}B), with the expected value \verb=null= and the new value \verb=I1=. If a CAS is successful, the insertion is completed and the Ctrie is in a state shown in Fig. \ref{f-tries}C. Otherwise, the insertion must be repeated.

Assume next that a key $k_2$ is inserted such that its hashcode prefix is different from that of $k_1$. By the hash trie invariants, $k_2$ should be next to $k_1$ in \verb=C1=. The thread that does the insertion first creates an updated version of \verb=C1= and then does a CAS at the \verb=I1.main= (Fig. \ref{f-tries}C) with the expected value of \verb=C1= and the updated node as the new value. Again, if the CAS is not successful, the insertion process is repeated. The Ctrie is now in the state shown in Fig. \ref{f-tries}D.

If some thread inserts a key $k_3$ with the same initial bits as $k_2$, the hash trie has to be extended with an additional level. The thread starts by creating a new node \verb=C2= of type \verb=CNode= containing both $k_2$ and $k_3$. It then creates a new node \verb=I2= and sets \verb=I2.main= to \verb=C2=. Finally, it creates a new updated version of \verb=C1= such that it points to the node \verb=I2= instead of the key $k_2$ and does a CAS at \verb=I1.main= (Fig. \ref{f-tries}D). We obtain a Ctrie shown in Fig. \ref{f-tries}E.

Assume now that a thread $T_1$ decides to remove $k_2$ from the Ctrie. It creates a new node \verb=C2'= from \verb=C2= which omits the key $k_2$. It then does a CAS on \verb=I2.main= to set it to \verb=C2'= (Fig. \ref{f-tries}E). As before, if the CAS is not successful, the operation is restarted. Otherwise, $k_2$ will no longer be in the trie -- concurrent operations will only see $k_1$ and $k_3$ in the trie, as shown in Fig. \ref{f-tries}F. However, the key $k_3$ could be moved further to the root - instead of being below the node \verb=C2=, it could be directly below the node \verb=C1=. In general, we want to ensure that the path from the root to a key is as short as possible. If we do not do this, we may end up with a lot of wasted space and an increased depth of the Ctrie.

For this reason, after having removed a key, a thread will attempt to contract the trie as much as possible. The thread $T_1$ that removed the key has to check whether or not there are less than 2 keys remaining within \verb=C2=. There is only a single key, so it can create a copy of \verb=C1= such that the key $k_3$ appears in place of the node \verb=I2= and then do a CAS at \verb=I1.main= (Fig. \ref{f-tries}F). However, this approach does not work. Assume there was another thread $T_2$ which decides to insert a new key below the node \verb=I2= just before $T_1$ does the CAS at \verb=I1.main=. The key inserted by $T_2$ is lost as soon as the CAS at \verb=I1.main= occurs.

To solve this, we relax the invariants of the data structure. We introduce a new type of a node - a tomb node. A tomb node is simply a node which holds a single key. No thread may modify a node of type \verb=INode= if it contains a tomb node. In our example, instead of directly modifying \verb=I1=, thread $T_1$ must first create a tomb node which contains the key $k_3$. It then does a CAS at \verb=I2.main= to set it to the tomb node. After having done this (Fig. \ref{f-tries}G), $T_1$ may create a contracted version of \verb=C1= and do a CAS at \verb=I1.main=, at which point we end up with a trie of an optimal size (Fig. \ref{f-tries}H). If some other thread $T_2$ attempts to modify \verb=I2= after it has been \textit{tombed}, then it must first do the same thing $T_1$ is attempting to do - move the key $k_3$ back below \verb=C2=, and only then proceed with its original operation. We call an \verb=INode= which points to a tomb node a \textit{tomb-inode}. We say that a tomb-inode in the example above is \textit{resurrected}.

If some thread decides to remove $k_1$, it proceeds as before. However, even though $k_3$ now remains the only key in \verb=C1= (Fig. \ref{f-tries}I), it does not get tombed. The reason for this is that we treat nodes directly below the root differently. If $k_3$ were next removed, the trie would end up in a state shown in Fig. \ref{f-tries}J, with the \verb=I1.main= set to \verb=null=. We call this type of an \verb=INode= a \textit{null-inode}.

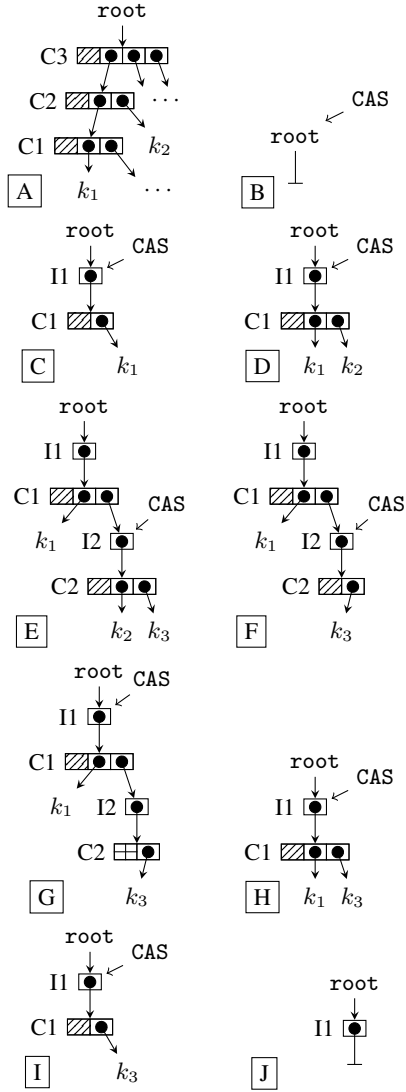
\begin{figure}[h]

\begin{center}
\begin{tabular}{ c c }

\begin{tikzpicture}
  [internal/.style={rectangle,draw},
   bmp/.style={rectangle,draw,pattern=north east lines},
   ]
  \node(root)     at ( 0.00, 0.00)                           {\verb=root=};

  \node(cn3)      at ( 0.00,-0.6)           [internal,minimum width=12mm,label=left:C3]        {};
  \node(cn3bmp)   at (-0.45,-0.6)           [bmp,minimum width=3mm]             {};
  \node(cn3arr1)  at (-0.15,-0.6)           [internal,minimum width=3mm]        {};
  \node(cn3arr2)  at ( 0.15,-0.6)           [internal,minimum width=3mm]        {};
  \node(cn3arr3)  at ( 0.45,-0.6)           [internal,minimum width=3mm]        {};
  \draw[-stealth](root)--(cn3);
  
  \node(cn2)      at (-0.30,-1.2)           [internal,minimum width=9mm,label=left:C2]        {};
  \node(cn2bmp)   at (-0.60,-1.2)           [bmp,minimum width=3mm]             {};
  \node(cn2arr1)  at (-0.30,-1.2)           [internal,minimum width=3mm]        {};
  \node(cn2arr2)  at ( 0.00,-1.2)           [internal,minimum width=3mm]        {};
  \draw[*-stealth,shorten <=-2.3pt](cn3arr1.center)--(cn2);
  
  \node(cn1)      at (-0.45,-1.8)           [internal,minimum width=9mm,label=left:C1]        {};
  \node(cn1bmp)   at (-0.75,-1.8)           [bmp,minimum width=3mm]             {};
  \node(cn1arr1)  at (-0.45,-1.8)           [internal,minimum width=3mm]        {};
  \node(cn1arr2)  at (-0.15,-1.8)           [internal,minimum width=3mm]        {};
  \draw[*-stealth,shorten <=-2.3pt](cn2arr1.center)--(cn1);
  
  \node(k1)       at (-0.45,-2.4)                                               {$k_1$};
  \draw[*-stealth,shorten <=-2.3pt](cn1arr1.center)--(k1);
  
  \node(k2)       at ( 0.50,-1.8)                                               {$k_2$};
  \draw[*-stealth,shorten <=-2.3pt](cn2arr2.center)--(k2);

  \node(more1)       at ( 0.60,-1.2)                                            {$\cdots$};
  \draw[*-stealth,shorten <=-2.3pt](cn3arr2.center)--(more1.north west);
  \draw[*-stealth,shorten <=-2.3pt](cn3arr3.center)--(more1.north);
  
  \node(more2)       at ( 0.50,-2.4)                                            {$\cdots$};
  \draw[*-stealth,shorten <=-2.3pt](cn1arr2.center)--(more2.north west);
  
  \node(lab)         at (-1.30,-2.4) [internal]     {A};
\end{tikzpicture}

&

\begin{tikzpicture}
  [internal/.style={rectangle,draw},
   leaf/.style={rectangle,draw,fill=black!20},
   ]
  \node(root)     at ( 0, 0)                           {\verb=root=};
  \node(null)     at ( 0,-0.75)                        {};
  \draw[-|](root)--(null);
  
  \node(cas)      at ( 1.0,0.5)                        {\verb=CAS=};
  \draw[->](cas)--(root);

  \node(lab)      at (-0.50,-0.75) [internal]     {B};
\end{tikzpicture}

\\

\begin{tikzpicture}
  [internal/.style={rectangle,draw},
   bmp/.style={rectangle,draw,pattern=north east lines},
   ]
  \node(root)     at ( 0.00, 0.00)                           {\verb=root=};
  \node(in1)      at ( 0.00,-0.60)          [internal,minimum width=3mm,label=left:I1]        {};
  \draw[-stealth](root)--(in1);
  
  \node(cn1)      at ( 0.00,-1.2)           [internal,minimum width=6mm,label=left:C1]        {};
  \node(cn1bmp)   at (-0.15,-1.2)           [bmp,minimum width=3mm]             {};
  \node(cn1arr1)  at ( 0.15,-1.2)           [internal,minimum width=3mm]        {};
  \draw[*-stealth,,shorten <=-2.3pt](in1.center)--(cn1);
  
  \node(k1)       at ( 0.50,-1.8)                                               {$k_1$};
  \draw[*-stealth,shorten <=-2.3pt](cn1arr1.center)--(k1);
  
  \node(cas)      at ( 0.8,-0.20)                              {\verb=CAS=};
  \draw[->, shorten >=2pt](cas)--(in1);

  \node(lab)      at (-0.70,-1.80) [internal]     {C};
\end{tikzpicture}

&

\begin{tikzpicture}
  [internal/.style={rectangle,draw},
   bmp/.style={rectangle,draw,pattern=north east lines},
   ]
  \node(root)     at ( 0.00, 0.00)                           {\verb=root=};
  \node(in1)      at ( 0.00,-0.60)          [internal,minimum width=3mm,label=left:I1]        {};
  \draw[-stealth](root)--(in1);
  
  \node(cn1)      at ( 0.00,-1.2)           [internal,minimum width=9mm,label=left:C1]        {};
  \node(cn1bmp)   at (-0.30,-1.2)           [bmp,minimum width=3mm]             {};
  \node(cn1arr1)  at ( 0.00,-1.2)           [internal,minimum width=3mm]        {};
  \node(cn1arr2)  at ( 0.30,-1.2)           [internal,minimum width=3mm]        {};
  \draw[*-stealth,shorten <=-2.3pt](in1.center)--(cn1);

  \node(k1)       at ( 0.00,-1.8)                                               {$k_1$};
  \draw[*-stealth,shorten <=-2.3pt](cn1arr1.center)--(k1);
  \node(k2)       at ( 0.50,-1.8)                                               {$k_2$};
  \draw[*-stealth,shorten <=-2.3pt](cn1arr2.center)--(k2);
  
  \node(cas)      at ( 0.8,-0.20)                              {\verb=CAS=};
  \draw[->, shorten >=2pt](cas)--(in1);

  \node(lab)      at (-0.70,-1.80) [internal]    {D};
\end{tikzpicture}

\\

\begin{tikzpicture}
  [internal/.style={rectangle,draw},
   bmp/.style={rectangle,draw,pattern=north east lines},
   ]
  \node(root)     at ( 0.00, 0.00)                           {\verb=root=};
  \node(in1)      at ( 0.00,-0.60)          [internal,minimum width=3mm,label=left:I1]        {};
  \draw[-stealth](root)--(in1);
  
  \node(cn1)      at ( 0.00,-1.2)           [internal,minimum width=9mm,label=left:C1]        {};
  \node(cn1bmp)   at (-0.30,-1.2)           [bmp,minimum width=3mm]             {};
  \node(cn1arr1)  at ( 0.00,-1.2)           [internal,minimum width=3mm]        {};
  \node(cn1arr2)  at ( 0.30,-1.2)           [internal,minimum width=3mm]        {};
  \draw[*-stealth,shorten <=-2.3pt](in1.center)--(cn1);
  
  \node(k1)       at (-0.50,-1.8)                                               {$k_1$};
  \draw[*-stealth,shorten <=-2.3pt](cn1arr1.center)--(k1);
  
  \node(in2)      at ( 0.50,-1.8)           [internal,minimum width=3mm,label=left:I2]        {};
  \draw[*-stealth,shorten <=-2.3pt](cn1arr2.center)--(in2);
  
  \node(cn2)      at ( 0.50,-2.4)           [internal,minimum width=9mm,label=left:C2]        {};
  \node(cn2bmp)   at ( 0.20,-2.4)           [bmp,minimum width=3mm]             {};
  \node(cn2arr1)  at ( 0.50,-2.4)           [internal,minimum width=3mm]        {};
  \node(cn2arr2)  at ( 0.80,-2.4)           [internal,minimum width=3mm]        {};
  \draw[*-stealth,shorten <=-2.3pt](in2.center)--(cn2);
  
  \node(k2)       at ( 0.50,-3.0)                                               {$k_2$};
  \draw[*-stealth,shorten <=-2.3pt](cn2arr1.center)--(k2);
  \node(k3)       at ( 1.00,-3.0)                                               {$k_3$};
  \draw[*-stealth,shorten <=-2.3pt](cn2arr2.center)--(k3);
  
  \node(cas)      at ( 1.1,-1.30)                              {\verb=CAS=};
  \draw[->, shorten >=2pt](cas)--(in2);
  
  \node(lab)      at (-0.70,-3.00) [internal]    {E};
\end{tikzpicture}

&

\begin{tikzpicture}
  [internal/.style={rectangle,draw},
   bmp/.style={rectangle,draw,pattern=north east lines},
   ]
  \node(root)     at ( 0.00, 0.00)                           {\verb=root=};
  \node(in1)      at ( 0.00,-0.60)          [internal,minimum width=3mm,label=left:I1]        {};
  \draw[-stealth](root)--(in1);
  
  \node(cn1)      at ( 0.00,-1.2)           [internal,minimum width=9mm,label=left:C1]        {};
  \node(cn1bmp)   at (-0.30,-1.2)           [bmp,minimum width=3mm]             {};
  \node(cn1arr1)  at ( 0.00,-1.2)           [internal,minimum width=3mm]        {};
  \node(cn1arr2)  at ( 0.30,-1.2)           [internal,minimum width=3mm]        {};
  \draw[*-stealth,shorten <=-2.3pt](in1.center)--(cn1);
  
  \node(k1)       at (-0.50,-1.8)                                               {$k_1$};
  \draw[*-stealth,shorten <=-2.3pt](cn1arr1.center)--(k1);
  
  \node(in2)      at ( 0.50,-1.8)           [internal,minimum width=3mm,label=left:I2]        {};
  \draw[*-stealth,shorten <=-2.3pt](cn1arr2.center)--(in2);
  
  \node(cn2)      at ( 0.50,-2.4)           [internal,minimum width=6mm,label=left:C2]        {};
  \node(cn2bmp)   at ( 0.35,-2.4)           [bmp,minimum width=3mm]             {};
  \node(cn2arr1)  at ( 0.65,-2.4)           [internal,minimum width=3mm]        {};
  \draw[*-stealth,shorten <=-2.3pt](in2.center)--(cn2);
  
  \node(k3)       at ( 0.50,-3.0)                                               {$k_3$};
  \draw[*-stealth,shorten <=-2.3pt](cn2arr1.center)--(k3);
  
  \node(cas)      at ( 1.1,-1.30)                              {\verb=CAS=};
  \draw[->, shorten >=2pt](cas)--(in2);

  \node(lab)      at (-0.70,-3.00) [internal]    {F};
\end{tikzpicture}

\\

\begin{tikzpicture}
  [internal/.style={rectangle,draw},
   bmp/.style={rectangle,draw,pattern=north east lines},
   tomb/.style={rectangle,draw},
   ]
  \node(root)     at ( 0.00, 0.00)                           {\verb=root=};
  \node(in1)      at ( 0.00,-0.60)          [internal,minimum width=3mm,label=left:I1]        {};
  \draw[-stealth](root)--(in1);
  
  \node(cn1)      at ( 0.00,-1.2)           [internal,minimum width=9mm,label=left:C1]        {};
  \node(cn1bmp)   at (-0.30,-1.2)           [bmp,minimum width=3mm]             {};
  \node(cn1arr1)  at ( 0.00,-1.2)           [internal,minimum width=3mm]        {};
  \node(cn1arr2)  at ( 0.30,-1.2)           [internal,minimum width=3mm]        {};
  \draw[*-stealth,shorten <=-2.3pt](in1.center)--(cn1);
  
  \node(k1)       at (-0.50,-1.8)                                               {$k_1$};
  \draw[*-stealth,shorten <=-2.3pt](cn1arr1.center)--(k1);
  
  \node(in2)      at ( 0.50,-1.8)           [internal,minimum width=3mm,label=left:I2]        {};
  \draw[*-stealth,shorten <=-2.3pt](cn1arr2.center)--(in2);
  
  \node(cn2)      at ( 0.50,-2.4)           [internal,minimum width=6mm,label=left:C2]        {};
  \node(cn2bmp)   at ( 0.35,-2.4)           [tomb,minimum width=3mm]            {};
  \draw[-](cn2bmp.south)--(cn2bmp.north);
  \draw[-](cn2bmp.east)--(cn2bmp.west);
  \node(cn2arr1)  at ( 0.65,-2.4)           [internal,minimum width=3mm]        {};
  \draw[*-stealth,shorten <=-2.3pt](in2.center)--(cn2);
  
  \node(k3)       at ( 0.50,-3.0)                                               {$k_3$};
  \draw[*-stealth,shorten <=-2.3pt](cn2arr1.center)--(k3);
  
  \node(cas)      at ( 0.7,-0.10)                              {\verb=CAS=};
  \draw[->, shorten >=2pt](cas)--(in1);

  \node(lab)      at (-0.70,-3.00) [internal]    {G};
\end{tikzpicture}

&

\begin{tikzpicture}
  [internal/.style={rectangle,draw},
   bmp/.style={rectangle,draw,pattern=north east lines},
   ]
  \node(root)     at ( 0.00, 0.00)                           {\verb=root=};
  \node(in1)      at ( 0.00,-0.60)          [internal,minimum width=3mm,label=left:I1]        {};
  \draw[-stealth](root)--(in1);
  
  \node(cn1)      at ( 0.00,-1.2)           [internal,minimum width=9mm,label=left:C1]        {};
  \node(cn1bmp)   at (-0.30,-1.2)           [bmp,minimum width=3mm]             {};
  \node(cn1arr1)  at ( 0.00,-1.2)           [internal,minimum width=3mm]        {};
  \node(cn1arr2)  at ( 0.30,-1.2)           [internal,minimum width=3mm]        {};
  \draw[*-stealth,shorten <=-2.3pt](in1.center)--(cn1);

  \node(k1)       at ( 0.00,-1.8)                                               {$k_1$};
  \draw[*-stealth,shorten <=-2.3pt](cn1arr1.center)--(k1);
  \node(k3)       at ( 0.50,-1.8)                                               {$k_3$};
  \draw[*-stealth,shorten <=-2.3pt](cn1arr2.center)--(k3);
  
  \node(cas)      at ( 0.8,-0.20)                              {\verb=CAS=};
  \draw[->, shorten >=2pt](cas)--(in1);
  
  \node(lab)      at (-0.70,-1.80) [internal]    {H};
\end{tikzpicture}

\\

\begin{tikzpicture}
  [internal/.style={rectangle,draw},
   bmp/.style={rectangle,draw,pattern=north east lines},
   ]
  \node(root)     at ( 0.00, 0.00)                           {\verb=root=};
  \node(in1)      at ( 0.00,-0.60)          [internal,minimum width=3mm,label=left:I1]        {};
  \draw[-stealth](root)--(in1);
  
  \node(cn1)      at ( 0.00,-1.2)           [internal,minimum width=6mm,label=left:C1]        {};
  \node(cn1bmp)   at (-0.15,-1.2)           [bmp,minimum width=3mm]             {};
  \node(cn1arr1)  at ( 0.15,-1.2)           [internal,minimum width=3mm]        {};
  \draw[*-stealth,shorten <=-2.3pt](in1.center)--(cn1);
  
  \node(k3)       at ( 0.50,-1.8)                                               {$k_3$};
  \draw[*-stealth,shorten <=-2.3pt](cn1arr1.center)--(k3);
  
  \node(cas)      at ( 0.8,-0.20)                              {\verb=CAS=};
  \draw[->, shorten >=2pt](cas)--(in1);

  \node(lab)      at (-0.70,-1.80) [internal]    {I};
\end{tikzpicture}

&

\begin{tikzpicture}
  [internal/.style={rectangle,draw},
   bmp/.style={rectangle,draw,pattern=north east lines},
   ]
  \node(root)     at ( 0.00, 0.00)                           {\verb=root=};
  \node(in1)      at ( 0.00,-0.60)          [internal,minimum width=3mm,label=left:I1]        {};
  \draw[-stealth](root)--(in1);
  
  \node(null)      at ( 0.00,-1.2)           [minimum width=6mm]                 {};
  \draw[*-|,shorten <=-2.3pt](in1.center)--(null);

  \node(lab)      at (-1.20,-1.20) [internal]    {J};
\end{tikzpicture}

\\

\end{tabular}
\end{center}

\caption{Hash trie and Ctrie examples}
\label{f-tries}
\end{figure}

\begin{figure}[h]
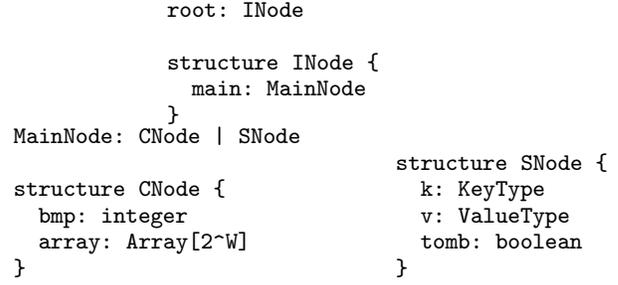

  \centering
  \begin{minipage}[b]{4 cm}
\begin{verbatim}
root: INode

structure INode {
  main: MainNode
}
\end{verbatim}
  \end{minipage}
  \begin{minipage}[b]{5 cm}
\begin{verbatim}
MainNode: CNode | SNode

structure CNode {
  bmp: integer
  array: Array[2^W]
}
\end{verbatim}
  \end{minipage}
  \begin{minipage}[b]{3 cm}
\begin{verbatim}
structure SNode {
  k: KeyType
  v: ValueType
  tomb: boolean
}
\end{verbatim}
  \end{minipage}

\caption{Types and data structures}
\label{f-types}
\end{figure}

\section{Algorithm}

We present the pseudocode of the algorithm in figures \ref{f-basic1}, \ref{f-basic2} and \ref{f-compress}. The pseudocode assumes C-like semantics of conditions in \textit{if} statements -- if the first condition in a conjunction fails, the second one is never evaluated. We use logical symbols for boolean expressions. The pseudocode also contains pattern matching constructs which are used to match a node against its type. All occurences of pattern matching can be trivially replaced with a sequence of \textit{if-then-else} statements -- we use pattern matching for conciseness. The colon (\verb=:=) in the pattern matching cases should be understood as \textit{has type}. The keyword \verb=def= denotes a procedure definition. Reads and compare-and-set instructions written in capitals are atomic -- they occur at one point in time. This is a high level pseudocode and might not be optimal in all cases -- the source code contains a more efficient implementation.

Operations start by reading the \verb=root= (lines \ref{read_topinsert}, \ref{read_topremove} and \ref{read_toplookup}). If the \verb=root= is \verb=null= then the trie is empty, so neither removal nor lookup finds a key. If the \verb=root= points to an \verb=INode= which is set to \verb=null= (as in Fig. \ref{f-tries}J), then the root is set back to just \verb=null= before repeating. In both the previous cases, an insertion will replace the \verb=root= reference with a new \verb=CNode= with the appropriate key.

If the \verb=root= is neither \verb=null= nor a null-inode then the node below the root inode is read (lines \ref{read_lookup}, \ref{read_insert} and \ref{read_remove}), and we proceed casewise. If the node pointed at by the inode is a \verb=CNode=, an appropriate entry in its array must be found. The method \verb=flagpos= computes the values \verb=flag= and \verb=pos= from the hashcode $hc$ of the key, bitmap $bmp$ of the cnode and the current level $lev$. The relevant \verb=flag= in the bitmap is defined as $(hc >> (W \cdot lev)) \odot ((1 << k) - 1)$, where $2^W$ is the length of the bitmap. The position \verb=pos= within the array is given by the expression $\#((flag - 1) \odot bmp)$, where $\#( \cdot )$ is the bitcount. The \verb=flag= is used to check if the appropriate branch is in the \verb=CNode= (lines \ref{flag_lookup}, \ref{flag_insert}, \ref{flag_remove}). If it is not, lookups and removes end, since the desired key is not in the Ctrie. An insert creates an updated copy of the current \verb=CNode= with the new key. If the branch is in the trie, \verb=pos= is used as an index into the array. If an inode is found, we repeat the operation recursively. If a key-value binding (an \verb=SNode=) is found, then a lookup compares the keys and returns the binding if they are the same. An insert operation will either replace the old binding if the keys are the same, or otherwise extend the trie below the \verb=CNode=. A remove operation compares the keys -- if they are the same it replaces the \verb=CNode= with its updated version without the key.

After a key was removed, the trie has to be contracted. A remove operation first attempts to create a tomb from the current \verb=CNode=. It first reads the node below the current inode to check if it is still a \verb=CNode=. It then calls \verb=toWeakTombed= which creates a \textit{weak tomb} from the given \verb=CNode=. A weak tomb is defined as follows. If the number of nodes below the \verb=CNode= that are not null-inodes is greater than 1, then it is the \verb=CNode= itself -- in this case we say that there is nothing to entomb. If the number of such nodes is 0, then the weak tomb is \verb=null=. Otherwise, if the single branch below the \verb=CNode= is a key-value binding or a tomb-inode (alternatively, a \textit{singleton}), the weak tomb is the tomb node with that binding. If the single branch below is another \verb=CNode=, a weak tomb is a copy of the current \verb=CNode= with the null-inodes removed.

The procedure \verb=tombCompress= continually tries to entomb the current \verb=CNode= until it finds out that there is nothing to entomb or it succeeds. The CAS in line \ref{cas_tomb} corresponds to the one in Fig. \ref{f-tries}F. If it succeeds and the weak tomb was either a \verb=null= or a tomb node, it will return \verb=true=, meaning that the parent node should be contracted. The contraction is done in \verb=contractParent=, which checks if the inode is still reachable from its parent and then contracts the \verb=CNode= below the parent - it removes the null-inode (line \ref{cas_contractnull}) or resurrects a tomb-inode into an \verb=SNode= (line \ref{cas_contractsingle}). The latter corresponds to the CAS in Fig. \ref{f-tries}G.

If any operation encounters a \verb=null= or a tomb node, it attempts to fix the Ctrie before proceeding, since the Ctrie is in a \textit{relaxed} state. A tomb node may have originated from a remove operation which will attempt to contract the tomb node at some time in the future. Rather than waiting for that remove operation to do its work, the current operation should do the work of contracting the tomb itself, so it will invoke the \verb=clean= operation on the parent inode. The \verb=clean= operation will attempt to exchange the \verb=CNode= below the parent inode with its compression. A \verb=CNode= compression is defined as follows -- if the \verb=CNode= has a single tomb node directly beneath, then it is that tomb node. Otherwise, the compression is the copy of the \verb=CNode= without the null-inodes (this is what the \verb=filtered= call in the \verb=toCompressed= procedure does) and with all the tomb-inodes resurrected to regular key nodes (this is what the \verb=map= and \verb=resurrect= calls do). Going back to our previous example, if in Fig. \ref{f-tries}G some other thread were to attempt to write to \verb=I2=, it would first do a \verb=clean= operation on the parent \verb=I1= of \verb=I2= -- it would contract the trie in the same way as the removal would have. After having fixed the Ctrie, the operation is repeated from the start.

\setlength\linenumbersep{2pt}

\begin{figure}[h]

  \centering
  \begin{minipage}[b]{6 cm}
\begin{alltt}
{\scriptsize
{\internallinenumbers{def insert(k, v)
  r = READ(root) {\customlabel{read_topinsert}{\LineNumber}}
  if r = null \(\vee\) isNullInode(r) \{ {\customlabel{check_topinsert}{\LineNumber}}
    scn = CNode(SNode(k, v, \(\bot\)))
    nr = INode(scn)
    if \(\lnot\)CAS(root, r, nr) insert(k, v) {\customlabel{cas_topinsert}{\LineNumber}}
  \} else if \(\lnot\)iinsert(r, k, v, 0, null)
    insert(k, v)

def remove(k)
  r = READ(root) {\customlabel{read_topremove}{\LineNumber}}
  if r = null return NOTFOUND
  else if isNullInode(r) \{ {\customlabel{check_topremove}{\LineNumber}}
    CAS(root, r, null) {\customlabel{cas_topremove}{\LineNumber}}
    return remove(k)
  \} else \{
    res = iremove(r, k, 0, null)
    if res \(\neq\) RESTART return res
    else remove(k)
  \}

def lookup(k)
  r = READ(root) {\customlabel{read_toplookup}{\LineNumber}}
  if r = null return NOTFOUND
  else if isNullInode(r) \{ {\customlabel{check_toplookup}{\LineNumber}}
    CAS(root, r, null) {\customlabel{cas_toplookup}{\LineNumber}}
    return lookup(k)
  \} else \{
    res = ilookup(r, k, 0, null)
    if res \(\neq\) RESTART return res
    else return lookup(k)
  \}

def ilookup(i, k, lev, parent)
  READ(i.main) match \{ {\customlabel{read_lookup}{\LineNumber}}
  case cn: CNode =>
    flag, pos = flagpos(k.hc, lev, cn.bmp)
    if cn.bmp \(\odot\) flag = 0 return NOTFOUND {\customlabel{flag_lookup}{\LineNumber}}
    cn.array(pos) match \{ {\customlabel{read_inlookup}{\LineNumber}}
    case sin: INode =>
      return ilookup(sin, k, lev + W, i)
    case sn: SNode \(\wedge\) \(\lnot\)sn.tomb =>
      if sn.k = k return sn.v {\customlabel{check_lookup}{\LineNumber}}
      else return NOTFOUND {\customlabel{notcheck_lookup}{\LineNumber}}
    \}
  case (sn: SNode \(\wedge\) sn.tomb) \(\vee\) null =>
    if parent \(\neq\) null clean(parent)
    return RESTART
  \}
}}}
\end{alltt}
  \end{minipage}

\caption{Basic operations I}
\label{f-basic1}
\end{figure}

\begin{figure}[h]

  \centering
  \begin{minipage}[t]{6 cm}
\begin{alltt}
{\scriptsize
{\internallinenumbers{def iinsert(i, k, v, lev, parent)
  READ(i.main) match \{ {\customlabel{read_insert}{\LineNumber}}
  case cn: CNode =>
    flag, pos = flagpos(k.hc, lev, cn.bmp)
    if cn.bmp \(\odot\) flag = 0 \{ {\customlabel{flag_insert}{\LineNumber}}
      nsn = SNode(k, v, \(\bot\))
      narr = cn.array.inserted(pos, nsn)
      ncn = CNode(narr, bmp | flag)
      return CAS(i.main, cn, ncn) {\customlabel{cas_insertcn}{\LineNumber}}
    \}
    cn.array(pos) match \{ {\customlabel{read_ininsert}{\LineNumber}}
    case sin: INode =>
      return iinsert(sin, k, v, lev + W, i)
    case sn: SNode \(\wedge\) \(\lnot\)sn.tomb =>
      nsn = SNode(k, v, \(\bot\))
      if sn.k = k \{
        ncn = cn.updated(pos, nsn)
        return CAS(i.main, cn, ncn) {\customlabel{cas_insertsn}{\LineNumber}}
      \} else \{
        nin = INode(CNode(sn, nsn, lev + W))
        ncn = cn.updated(pos, nin)
        return CAS(i.main, cn, ncn) {\customlabel{cas_insertin}{\LineNumber}}
      \}
    \}
  case (sn: SNode \(\wedge\) sn.tomb) \(\vee\) null =>
    if parent \(\neq\) null clean(parent)
    return \(\bot\)
  \}
  
def iremove(i, k, lev, parent)
  READ(i.main) match \{ {\customlabel{read_remove}{\LineNumber}}
  case cn: CNode =>
    flag, pos = flagpos(k.hc, lev, cn.bmp)
    if cn.bmp \(\odot\) flag = 0 return NOTFOUND {\customlabel{flag_remove}{\LineNumber}}
    res = cn.array(pos) match \{ {\customlabel{read_inremove}{\LineNumber}}
    case sin: INode =>
      return iremove(sin, k, lev + W, i)
    case sn: SNode \(\wedge\) \(\lnot\)sn.tomb =>
      if sn.k = k \{
        narr = cn.array.removed(pos)
        ncn = CNode(narr, bmp ^ flag)
        if cn.array.length = 1 ncn = null
        if CAS(i.main, cn, ncn) return sn.v {\customlabel{cas_remove}{\LineNumber}}
        else return RESTART
      \} else return NOTFOUND {\customlabel{notcheck_remove}{\LineNumber}}
    \}
    if res = NOTFOUND \(\vee\) res = RESTART return res
    if parent ne null \(\wedge\) tombCompress()
      contractParent(parent, in, k.hc, lev - W)
  case (sn: SNode \(\wedge\) sn.tomb) \(\vee\) null =>
    if parent \(\neq\) null clean(parent)
    return RESTART
  \}
}}}
\end{alltt}
  \end{minipage}

\caption{Basic operations II}
\label{f-basic2}
\end{figure}

\begin{figure}[h]

  \centering
  \begin{minipage}[b]{6 cm}
\begin{alltt}
{\scriptsize
{\internallinenumbers{def toCompressed(cn)
  num = bit#(cn.bmp)
  if num = 1 \(\wedge\) isTombInode(cn.array(0))
    return cn.array(0).main
  ncn = cn.filtered(_.main \(\neq\) null)
  rarr = ncn.array.map(resurrect(_))
  if bit#(ncn.bmp) > 0
    return CNode(rarr, ncn.bmp)
  else return null

def toWeakTombed(cn)
  farr = cn.array.filtered(_.main \(\neq\) null)
  nbmp = cn.bmp.filtered(_.main \(\neq\) null)
  if farr.length > 1 return cn
  if farr.length = 1
    if isSingleton(farr(0))
      return farr(0).tombed
    else CNode(farr, nbmp)
  return null

def clean(i)
  m = READ(i.main) {\customlabel{read_clean}{\LineNumber}}
  if m \(\in\) CNode
    CAS(i.main, m, toCompressed(m)) {\customlabel{cas_clean}{\LineNumber}}

def tombCompress(i)
  m = READ(i.main) {\customlabel{read_tomb}{\LineNumber}}
  if m \(\not\in\) CNode return \(\bot\)
  mwt = toWeakTombed(m)
  if m = mwt return \(\bot\)
  if CAS(i.main, m, mwt) mwt match \{ {\customlabel{cas_tomb}{\LineNumber}}
    case null \(\vee\) (sn: SNode \(\wedge\) sn.tomb) =>
      return \(\top\)
    case _ => return \(\bot\)
  \} else return tombCompress()

def contractParent(parent, i, hc, lev)
  m, pm = READ(i.main), READ(parent.main) {\customlabel{read_contract}{\LineNumber}}
  pm match \{
  case cn: CNode =>
    flag, pos = flagpos(k.hc, lev, cn.bmp)
    if bmp \(\odot\) flag = 0 return
    sub = cn.array(pos) {\customlabel{read_incontract}{\LineNumber}}
    if sub \(\neq\) i return
    if m = null \{
      ncn = cn.removed(pos)
      if \(\lnot\)CAS(parent.main, cn, ncn) {\customlabel{cas_contractnull}{\LineNumber}}
        contractParent(parent, i, hc, lev)
    \} else if isSingleton(m) \{
      ncn = cn.updated(pos, m.untombed)
      if \(\lnot\)CAS(parent.main, cn, ncn) {\customlabel{cas_contractsingle}{\LineNumber}}
        contractParent(parent, i, hc, lev)
    \}
  case _ => return
  \}
}}}
\end{alltt}
  \end{minipage}

\caption{Compression operations}
\label{f-compress}
\end{figure}

\section{Correctness}

As illustrated by the examples in the previous section, designing a correct lock-free algorithm is not straightforward. One of the reasons for this is that all possible interleavings of steps of different threads executing the operations have to be considered. For brevity, this section gives only the outline of the correctness proof -- the complete proof is given in the appendix. There are three main criteria for correctness. \textit{Safety} means that the Ctrie corresponds to some abstract set of keys and that all operations change the corresponding abstract set of keys consistently. An operation is \textit{linearizable} if any external observer can only observe the operation as if it took place instantaneously at some point between its invocation and completion \cite{artofmulti08} \cite{maurice-linear}. \textit{Lock-freedom} means that if some number of threads execute operations concurrently, then after a finite number of steps some operation must complete \cite{artofmulti08}.

We assume that the Ctrie has a branching factor $2^W$. Each node in the Ctrie is identified by its type, level in the Ctrie $l$ and the hashcode prefix $p$. The hashcode prefix is the sequence of branch indices that have to be followed from the root in order to reach the node. For a cnode $cn_{l,p}$ and a key $k$ with the hashcode $h = r_0 \cdot r_1 \cdots r_n$, we denote $cn.sub(k)$ as the branch with the index $r_l$ or $null$ if such a branch does not exist. We define the following invariants:

\small
\begin{description}
\item[INV1] For every inode $in_{l,p}$, $in_{l,p}.main$ is a cnode $cn_{l,p}$, a tombed snode $sn\dagger$ or $null$.
\item[INV2] For every cnode the length of the array is equal to the bitcount in the bitmap.
\item[INV3] If a flag $i$ in the bitmap of $cn_{l,p}$ is set, then corresponding array entry contains an inode $in_{l + W, p \cdot r}$ or an snode.
\item[INV4] If an entry in the array in $cn_{l,p}$ contains an snode $sn$, then $p$ is the prefix of the hashcode $sn.k$.
\item[INV5] If an inode $in_{l,p}$ contains an snode $sn$, then $p$ is the prefix of the hashcode $sn.k$.
\end{description}
\normalsize

We say that the Ctrie is \textit{valid} if and only if the invariants hold.
The relation $hasKey(node, x)$ holds if and only if the key $x$ is within an snode reachable from $node$.
A valid Ctrie is \textit{consistent} with an abstract set $\mathbb{A}$ if and only if $\forall k \in \mathbb{A}$ the relation $hasKey(root, k)$ holds and $\forall k \notin \mathbb{A}$ it does not.
A Ctrie lookup is \textit{consistent} with the abstract set semantics if and only if it finds the keys in the abstract set and does not find other keys. A Ctrie insertion or removal is \textit{consistent} with the abstract set semantics if and only if it produces a new Ctrie consistent with a new abstract set with or without the given key, respectively.

\begin{lemma}
If an inode $in$ is either a null-inode or a tomb-inode at some time $t_0$ then $\forall t > t_0$ $in.main$ is not written. We refer to such inodes as \textit{nonlive}.
\end{lemma}

\begin{lemma}
Cnodes and snodes are immutable -- once created, they do not change the value of their fields.
\end{lemma}

\begin{lemma}
Invariants INV1-3 always hold.
\end{lemma}

\begin{lemma}
If a CAS instruction makes an inode $in$ unreachable from its parent at some time $t_0$, then $in$ is nonlive at $t_0$.
\end{lemma}

\begin{lemma}
Reading a $cn$ such that $cn.sub(k) = sn$ and $k = sn.k$ at some time $t_0$ means that $hasKey(root, k)$ holds at $t_0$.
\end{lemma}

For a given Ctrie, we say that the longest path for a hashcode $h = r_0 \cdot r_1 \cdots r_n$, $length(r_i) = W$, is the path from the root to a leaf such that at each cnode $cn_{i,p}$ the branch with the index $r_i$ is taken.

\begin{lemma}
Assume that the Ctrie is an valid state. Then every longest path ends with an snode, cnode or $null$.
\end{lemma}

\begin{lemma}
Assume that a cnode $cn$ is read from $in_{l,p}.main$ at some time $t_0$ while searching for a key $k$. If $cn.sub(k) = null$ then $hasKey(root, k)$ is not in the Ctrie at $t_0$.
\end{lemma}

\begin{lemma}
Assume that the algorithm is searching for a key $k$ and that an snode $sn$ is read from $cn.array(i)$ at some time $t_0$ such that $sn.k \neq k$. Then the relation $hasKey(root, k)$ does not hold at $t_0$.
\end{lemma}

\begin{lemma}\label{l-maintext-fastening}
1. Assume that one of the CAS in lines \ref{cas_insertcn} and \ref{cas_insertin} succeeds at time $t_1$ after $in.main$ was read in line \ref{read_insert} at time $t_0$. Then $\forall t, t_0 \le t < t_1$, relation $hasKey(root, k)$ does not hold.

2. Assume that the CAS in lines \ref{cas_insertsn} succeeds at time $t_1$ after $in.main$ was read in line \ref{read_insert} at time $t_0$. Then $\forall t, t_0 \le t < t_1$, relation $hasKey(root, k)$ holds.

3. Assume that the CAS in line \ref{cas_remove} succeeds at time $t_1$ after $in.main$ was read in line \ref{read_remove} at time $t_0$. Then $\forall t, t_0 \le t < t_1$, relation $hasKey(root, k)$ holds.
\end{lemma}

\begin{lemma}\label{l-maintext-consmodif}
Assume that the Ctrie is valid and consistent with some abstract set $\mathbb{A}$ $\forall t, t_1 - \delta < t < t_1$. CAS instructions from lemma \ref{l-maintext-fastening} induce a change into a valid state which is consistent with the abstract set semantics.
\end{lemma}

\begin{lemma}\label{l-maintext-consclean}
Assume that the Ctrie is valid and consistent with some abstract set $\mathbb{A}$ $\forall t, t_1 - \delta < t < t_1$. If one of the operations $clean$, $tombCompress$ or $contractParent$ succeeds with a CAS at $t_1$, the Ctrie will remain valid and consistent with the abstract set $\mathbb{A}$ at $t_1$.
\end{lemma}

\begin{corollary}
Invariants INV4,5 always hold due to lemmas \ref{l-maintext-consmodif} and \ref{l-maintext-consclean}.
\end{corollary}

\begin{theorem}[Safety]
At all times $t$, a Ctrie is in a valid state $\mathbb{S}$, consistent with some abstract set $\mathbb{A}$. All Ctrie operations are consistent with the semantics of the abstract set $\mathbb{A}$.
\end{theorem}

\begin{theorem}[Linearizability]
Ctrie operations are linearizable.
\end{theorem}


\begin{lemma}
If a CAS that does not cause a consistency change in one of the lines \ref{cas_insertcn}, \ref{cas_insertsn}, \ref{cas_insertin}, \ref{cas_clean}, \ref{cas_tomb}, \ref{cas_contractnull} or \ref{cas_contractsingle} fails at some time $t_1$, then there has been a state (configuration) change since the time $t_0$ when a respective read in one of the lines \ref{read_insert}, \ref{read_insert}, \ref{read_insert}, \ref{read_clean}, \ref{read_tomb}, \ref{read_contract} or \ref{read_contract} occured.
\end{lemma}

\begin{lemma}
In each operation there is a finite number of execution steps between consecutive CAS instructions.
\end{lemma}

\begin{corollary}
There is a finite number of execution steps between two state changes. This does not imply that there is a finite number of execution steps between two operations. A state change is not necessarily a consistency change.
\end{corollary}

We define the \textbf{total path length} $d$ as the sum of the lengths of all the paths from the root to some leaf.
Assume the Ctrie is in a valid state. Let $n$ be the number of reachable null-inodes in this state, $t$ the number of reachable tomb-inodes, $l$ the number of live inodes, $r$ the number of single tips of any length and $d$ the total path length. We denote the state of the Ctrie as $\mathbb{S}_{n,t,l,r,d}$.
We call the state $\mathbb{S}_{0,0,l,r,d}$ the \textbf{clean} state.

\begin{lemma}
Observe all CAS instructions which never cause a consistency change and assume they are successful.
Assuming there was no state change since reading $in$ prior to calling $clean$, the CAS in line \ref{cas_clean} changes the state of the Ctrie from the state $\mathbb{S}_{n,t,l,r,d}$ to either $\mathbb{S}_{n+j,t,l,r-1,d-1}$ where $r > 0$, $j \in \{ 0, 1 \}$ and $d \ge 1$, or to $\mathbb{S}_{n-k,t-j,l,r,d' \le d}$ where $k \ge 0$, $j \ge 0$, $k + j > 0$, $n \ge k$ and $t \ge j$.
Furthermore, the CAS in line \ref{cas_topremove} changes the state of the Ctrie from $\mathbb{S}_{1,0,0,0,1}$ to $\mathbb{S}_{0,0,0,0,0}$.
The CAS in line \ref{cas_toplookup} changes the state from $\mathbb{S}_{1,0,0,0,1}$ to $\mathbb{S}_{0,0,0,0,0}$.
The CAS in line \ref{cas_tomb} changes the state from $\mathbb{S}_{n,t,l,r,d}$ to either $\mathbb{S}_{n+j,t,l,r-1,d-j}$ where $r > 0$, $j \in \{ 0, 1 \}$ and $d \ge j$, or to $\mathbb{S}_{n-k,t,l,r,d' \le d}$ where $k > 0$ and $n \ge k$.
The CAS in line \ref{cas_contractnull} changes the state from $\mathbb{S}_{n,t,l,r,d}$ to $\mathbb{S}_{n-1,t,l,r+j,d-1}$ where $n > 0$ and $j \ge 0$.
The CAS in line \ref{cas_contractsingle} changes the state from $\mathbb{S}_{n,t,l,r}$ to $\mathbb{S}_{n,t-1,l,r+j,d-1}$ where $j \ge 0$.
\end{lemma}

\begin{lemma}
If the Ctrie is in a clean state and $n$ threads are executing operations on it, then some thread will execute a successful CAS resulting in a consistency change after a finite number of execution steps.
\end{lemma}

\begin{theorem}[Lock-freedom]
Ctrie operations are lock-free.
\end{theorem}

\section{Experiments}

We show benchmark results in Fig. \ref{f-4xi7}. All the measurements were performed on a quad-core 2.67 GHz i7 processor with hyperthreading. We followed established performance measurement methodologies \cite{perfeval}. We compare the performance of Ctries against that of \verb=ConcurrentHashMap= and \verb=ConcurrentSkipListMap= \cite{dougleahome} \cite{conchashtable-maged} data structures from the Java standard library.

In the first experiment, we insert a total of $N$ elements into the data structures. The insertion is divided equally among $P$ threads, where $P$ ranges from $1$ to $8$. The results are shown in Fig. \ref{f-4xi7}A-D. Ctries outperform concurrent skip lists for $P=1$ (Fig. \ref{f-4xi7}A). We argue that this is due to a fewer number of indirections in the Ctrie data structure. A concurrent skip list roughly corresponds to a balanced binary tree which has a branching factor $2$. Ctries normally have a branching factor $32$, thus having a much lower depth. A lower depth means less indirections and consequently fewer cache misses when searching the Ctrie.

We can also see that the Ctrie sometimes outperforms a concurrent hash table for $P=1$. The reason is that the hash table has a fixed size and is resized once the load factor is reached -- a new table is allocated and all the elements from the previous hash table have to be copied into the new hash table. More importantly, this implementation uses a global write lock during the resize phase -- other threads adding new elements into the table have to wait until the resizal completes. This problem is much more apparent in Fig. \ref{f-4xi7}B where $P=8$. Fig. \ref{f-4xi7}C,D show how the insertion scales for the number of elements $N=200k$ and $N=1M$, respectively. Due to the use of hyperthreading on the i7, we do not get significant speedups when $P > 4$ for these data structures.

We next measure the performance for the remove operation (Fig. \ref{f-4xi7}E-H). Each data structure starts with $N$ elements and then emptied concurrently by $P$ threads. The keys being removed are divided equally among the threads. For $P=1$ Ctries are clearly outperformed by both other data structures. However, it should be noted that concurrent hash table does not shrink once the number of keys becomes much lower than the table size. This is space-inefficient -- a hash table contains many elements at some point during the runtime of the application will continue to use the memory it does not need until the application ends. The slower Ctrie performance seen in Fig. \ref{f-4xi7}E for $P=1$ is attributed to the additional work the remove operation does to keep the Ctrie compact. However, Fig. \ref{f-4xi7}F shows that the Ctrie remove operation scales well for $P=8$, as it outperforms both skip list and hash table removals. This is also apparent in Fig. \ref{f-4xi7}G,H.

In the next experiment, we populate all the data structures with $N$ elements and then do a lookup for every element once. The set of elements to be looked up is divided equally among $P$ threads. From Fig. \ref{f-4xi7}I-L it is apparent that concurrent hash tables have a much more efficient lookups than other data structures. This is not surprising since they are a flat data structure -- a lookup typically consists of a single read in the table, possibly followed by traversing the collision chain within the bucket. Although a Ctrie lookup outperforms a concurrent skip list when $P=8$, it still has to traverse more indirections than a hash table.

Finally, we do a series of benchmarks with both lookups and insertions to determine the percentage of lookups for which the concurrent hash table performance equals that of concurrent tries. Our test inserts new elements into the data structures using $P$ threads. A total of $N$ elements are inserted. After each insert, a lookup for a random element is performed $r$ times, where $r$ is the ratio of lookups per insertion. Concurrent skip lists scaled well in these tests but had low absolute performance, so they are excluded from the graphs for clarity. When using $P=2$ threads, the ratio where the running time is equal for both concurrent hash tables and concurrent tries is $r=2$. When using $P=4$ threads this ratio is $r=5$ and for $P=8$ the ratio is $r=9$. As the number of threads increases, more opportunity for parallelism is lost during the resizal phase in concurrent hash tables, hence the ratio increases. This is shown in Fig. \ref{f-4xi7-insert-lookup}A-C. In the last benchmark (Fig. \ref{f-4xi7-insert-lookup}D) we preallocate the array for the concurrent hash table to avoid resizal phases -- in this case the hash table outperforms the concurrent trie. The performance gap decreases as the number of threads approaches $P=8$. The downside is that a large amount of memory has to be used for the hash table and the size needs to be known in advance.

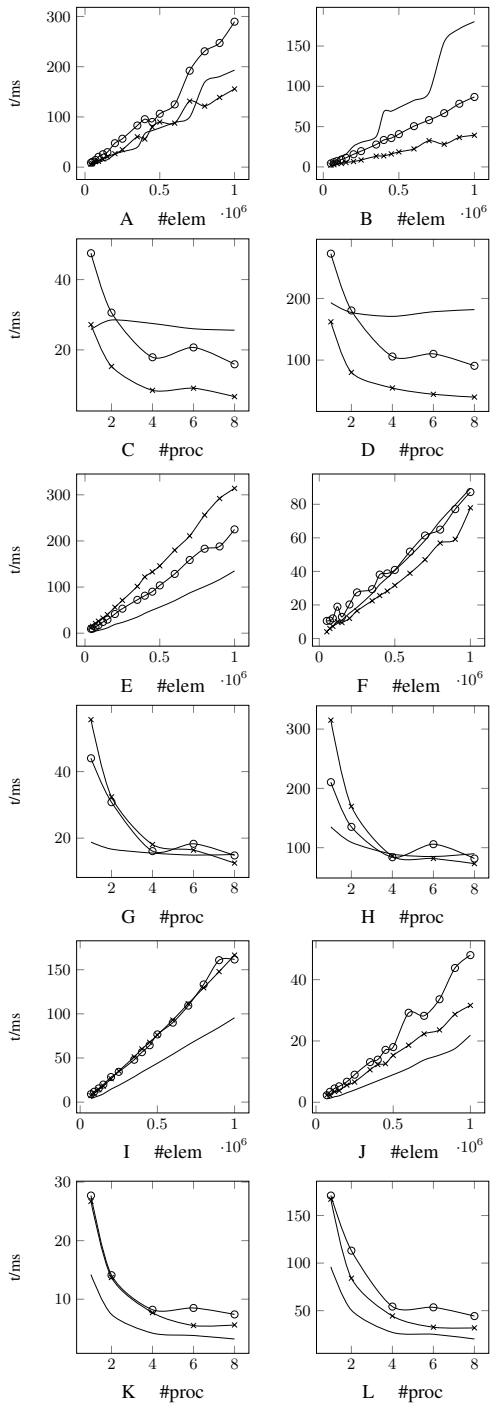
\begin{figure}[h]

\begin{center}
\begin{tabular}{ c c }
\begin{tikzpicture}[scale=0.67]

\begin{axis}[xlabel=\large A \mbox{ } \mbox{ } \#elem,ylabel=t/ms,height=5cm,width=5cm]

\addplot[smooth]
plot coordinates {
  (40000, 4)
  (50000, 5.5)
  (70000, 6.9)
  (90000, 8.5)
  (120000, 12.7)
  (150000, 14.9)
  (200000, 25.7)
  (250000, 30.7)
  (350000, 39.6)
  (400000, 66.7)
  (450000, 72.8)
  (500000, 78)
  (600000, 89)
  (700000, 99.6)
  (800000, 168.9)
  (900000, 180.3)
  (1000000, 193.2)
};

\addplot[smooth,mark=o]
plot coordinates {
  (40000, 8.9)
  (50000, 10.7)
  (70000, 14.4)
  (90000, 20.6)
  (120000, 25.7)
  (150000, 29.9)
  (200000, 47.7)
  (250000, 56.6)
  (350000, 82.9)
  (400000, 95.5)
  (450000, 90)
  (500000, 106.3)
  (600000, 124.9)
  (700000, 192)
  (800000, 230.7)
  (900000, 247.2)
  (1000000, 289.6)
};

\addplot[smooth,mark=x]
plot coordinates {
  (40000, 5.9)
  (50000, 7.2)
  (70000, 10.9)
  (90000, 11.8)
  (120000, 16)
  (150000, 20.8)
  (200000, 26.8)
  (250000, 34.8)
  (350000, 60.4)
  (400000, 56)
  (450000, 79.6)
  (500000, 89.9)
  (600000, 87.6)
  (700000, 131.4)
  (800000, 121.5)
  (900000, 138.8)
  (1000000, 155.7)
};

\end{axis}

\end{tikzpicture}

&

\begin{tikzpicture}[scale=0.67]

\begin{axis}[xlabel=\large B \mbox{ } \mbox{ } \#elem,height=5cm,width=5cm]

\addplot[smooth]
plot coordinates {
  (50000, 6.3)
  (70000, 7.5)
  (90000, 8.7)
  (120000, 13.3)
  (150000, 15.7)
  (200000, 25.7)
  (250000, 30.5)
  (350000, 37.7)
  (400000, 68.3)
  (450000, 69.1)
  (500000, 73.6)
  (600000, 82.6)
  (700000, 92.2)
  (800000, 155.1)
  (900000, 171.4)
  (1000000, 180.2)
};

\addplot[smooth,mark=o]
plot coordinates {
  (50000, 4.3)
  (70000, 5.9)
  (90000, 7.2)
  (120000, 9.2)
  (150000, 11.3)
  (200000, 15.7)
  (250000, 19.6)
  (350000, 27.8)
  (400000, 33.4)
  (450000, 35.8)
  (500000, 40.7)
  (600000, 50.6)
  (700000, 58.1)
  (800000, 66.7)
  (900000, 78.4)
  (1000000, 86.8)
};

\addplot[smooth,mark=x]
plot coordinates {
  (50000, 2.2)
  (70000, 2.8)
  (90000, 4.3)
  (120000, 4.3)
  (150000, 5.5)
  (200000, 6.7)
  (250000, 8.6)
  (350000, 13.3)
  (400000, 13.7)
  (450000, 16.1)
  (500000, 18.6)
  (600000, 22.4)
  (700000, 32.6)
  (800000, 28.2)
  (900000, 36.6)
  (1000000, 39.3)
};

\end{axis}

\end{tikzpicture}

\\

\begin{tikzpicture}[scale=0.67]

\begin{axis}[xlabel=\large C \mbox{ } \mbox{ } \#proc,ylabel=t/ms,height=5cm,width=5cm]

\addplot[smooth]
plot coordinates {
  (1, 25.8)
  (2, 28.5)
  (4, 27.5)
  (6, 26)
  (8, 25.6)
};

\addplot[smooth,mark=o]
plot coordinates {
  (1, 47.5)
  (2, 30.6)
  (4, 17.9)
  (6, 20.7)
  (8, 15.9)
};

\addplot[smooth,mark=x]
plot coordinates {
  (1, 27.2)
  (2, 15.3)
  (4, 8.5)
  (6, 9.1)
  (8, 6.7)
};

\end{axis}

\end{tikzpicture}

&

\begin{tikzpicture}[scale=0.67]

\begin{axis}[xlabel=\large D \mbox{ } \mbox{ } \#proc,height=5cm,width=5cm]

\addplot[smooth]
plot coordinates {
  (1, 193.1)
  (2, 177.2)
  (4, 170.9)
  (6, 178.5)
  (8, 181.9)
};

\addplot[smooth,mark=o]
plot coordinates {
  (1, 272.7)
  (2, 180.5)
  (4, 105.8)
  (6, 110.1)
  (8, 90.7)
};

\addplot[smooth,mark=x]
plot coordinates {
  (1, 162.1)
  (2, 80.1)
  (4, 54.8)
  (6, 44.5)
  (8, 40)
};

\end{axis}

\end{tikzpicture}

\\


\begin{tikzpicture}[scale=0.67]

\begin{axis}[xlabel=\large E \mbox{ } \mbox{ } \#elem,ylabel=t/ms,height=5cm,width=5cm]

\addplot[smooth]
plot coordinates {
  (40000, 3.2)
  (50000, 3)
  (70000, 4.5)
  (90000, 6.1)
  (120000, 8.5)
  (150000, 11.5)
  (200000, 19.1)
  (250000, 23.4)
  (350000, 35)
  (400000, 43.1)
  (450000, 50.1)
  (500000, 56.6)
  (600000, 70.5)
  (700000, 87.3)
  (800000, 101.1)
  (900000, 116.2)
  (1000000, 135.1)
};

\addplot[smooth,mark=o]
plot coordinates {
  (40000, 9.4)
  (50000, 9.7)
  (70000, 13.9)
  (90000, 17.4)
  (120000, 24.3)
  (150000, 29.4)
  (200000, 41.6)
  (250000, 53.2)
  (350000, 72.3)
  (400000, 81.3)
  (450000, 90)
  (500000, 103.6)
  (600000, 128.5)
  (700000, 158.7)
  (800000, 183.1)
  (900000, 188.1)
  (1000000, 225.1)
};

\addplot[smooth,mark=x]
plot coordinates {
  (40000, 11.9)
  (50000, 15)
  (70000, 21)
  (90000, 26)
  (120000, 32.7)
  (150000, 40.3)
  (200000, 56.2)
  (250000, 71.2)
  (350000, 101.5)
  (400000, 122)
  (450000, 133)
  (500000, 146)
  (600000, 180)
  (700000, 211)
  (800000, 256)
  (900000, 292)
  (1000000, 314)
};

\end{axis}

\end{tikzpicture}

&

\begin{tikzpicture}[scale=0.67]

\begin{axis}[xlabel=\large F \mbox{ } \mbox{ } \#elem,height=5cm,width=5cm]

\addplot[smooth]
plot coordinates {
  (50000, 4.5)
  (70000, 5.7)
  (90000, 7.4)
  (120000, 8.9)
  (150000, 11)
  (200000, 15.1)
  (250000, 18.6)
  (350000, 26.8)
  (400000, 32.2)
  (450000, 36.2)
  (500000, 40.5)
  (600000, 49.2)
  (700000, 58.5)
  (800000, 70.1)
  (900000, 79.7)
  (1000000, 89.5)
};

\addplot[smooth,mark=o]
plot coordinates {
  (50000, 10.5)
  (70000, 10.7)
  (90000, 12)
  (120000, 19)
  (150000, 12.9)
  (200000, 20.3)
  (250000, 27.5)
  (350000, 29.4)
  (400000, 38.1)
  (450000, 38.8)
  (500000, 40.9)
  (600000, 51.8)
  (700000, 61.4)
  (800000, 64.9)
  (900000, 77.1)
  (1000000, 87.3)
};

\addplot[smooth,mark=x]
plot coordinates {
  (50000, 4)
  (70000, 6)
  (90000, 7.3)
  (120000, 10.1)
  (150000, 9.5)
  (200000, 12)
  (250000, 16.5)
  (350000, 22.5)
  (400000, 25.6)
  (450000, 28.3)
  (500000, 31.7)
  (600000, 38.9)
  (700000, 47)
  (800000, 56.9)
  (900000, 59.2)
  (1000000, 77.9)
};

\end{axis}

\end{tikzpicture}

\\

\begin{tikzpicture}[scale=0.67]

\begin{axis}[xlabel=\large G \mbox{ } \mbox{ } \#proc,ylabel=t/ms,height=5cm,width=5cm]

\addplot[smooth]
plot coordinates {
  (1, 18.8)
  (2, 16.7)
  (4, 15.5)
  (6, 14.9)
  (8, 15.1)
};

\addplot[smooth,mark=o]
plot coordinates {
  (1, 44)
  (2, 30.8)
  (4, 16.1)
  (6, 18.3)
  (8, 14.8)
};

\addplot[smooth,mark=x]
plot coordinates {
  (1, 55.6)
  (2, 32.4)
  (4, 18.1)
  (6, 16.4)
  (8, 12.5)
};

\end{axis}

\end{tikzpicture}

&

\begin{tikzpicture}[scale=0.67]

\begin{axis}[xlabel=\large H \mbox{ } \mbox{ } \#proc,height=5cm,width=5cm]

\addplot[smooth]
plot coordinates {
  (1, 135.1)
  (2, 109)
  (4, 89.2)
  (6, 85.1)
  (8, 89.9)
};

\addplot[smooth,mark=o]
plot coordinates {
  (1, 210.5)
  (2, 135.1)
  (4, 83.9)
  (6, 105.9)
  (8, 81.7)
};

\addplot[smooth,mark=x]
plot coordinates {
  (1, 315)
  (2, 169.7)
  (4, 85.7)
  (6, 81.8)
  (8, 73.1)
};

\end{axis}

\end{tikzpicture}

\\


\begin{tikzpicture}[scale=0.67]

\begin{axis}[xlabel=\large I \mbox{ } \mbox{ } \#elem,ylabel=t/ms,height=5cm,width=5cm]

\addplot[smooth]
plot coordinates {
  (70000, 4)
  (90000, 5.3)
  (120000, 7.1)
  (150000, 9.3)
  (200000, 14.7)
  (250000, 19)
  (350000, 28.7)
  (400000, 34.1)
  (450000, 39)
  (500000, 44)
  (600000, 54.2)
  (700000, 64.9)
  (800000, 74.9)
  (900000, 84.7)
  (1000000, 95.6)
};

\addplot[smooth,mark=o]
plot coordinates {
  (70000, 9.1)
  (90000, 11.7)
  (120000, 15.5)
  (150000, 19.9)
  (200000, 28.4)
  (250000, 34.3)
  (350000, 48)
  (400000, 56.6)
  (450000, 64.4)
  (500000, 76.7)
  (600000, 90)
  (700000, 109.3)
  (800000, 133.3)
  (900000, 160.8)
  (1000000, 161.6)
};

\addplot[smooth,mark=x]
plot coordinates {
  (70000, 7.5)
  (90000, 9.8)
  (120000, 14.1)
  (150000, 18.2)
  (200000, 26.5)
  (250000, 35.2)
  (350000, 51.2)
  (400000, 60.2)
  (450000, 67.8)
  (500000, 76)
  (600000, 93)
  (700000, 111.5)
  (800000, 129.6)
  (900000, 147.8)
  (1000000, 166.5)
};

\end{axis}

\end{tikzpicture}

&

\begin{tikzpicture}[scale=0.67]

\begin{axis}[xlabel=\large J \mbox{ } \mbox{ } \#elem,height=5cm,width=5cm]

\addplot[smooth]
plot coordinates {
  (70000, 1.2)
  (90000, 1.3)
  (120000, 1.8)
  (150000, 2.1)
  (200000, 3.1)
  (250000, 4)
  (350000, 6)
  (400000, 7)
  (450000, 8)
  (500000, 9)
  (600000, 11.1)
  (700000, 13.8)
  (800000, 15.4)
  (900000, 17.5)
  (1000000, 21.9)
};

\addplot[smooth,mark=o]
plot coordinates {
  (70000, 2.3)
  (90000, 3.4)
  (120000, 4.5)
  (150000, 5.2)
  (200000, 6.7)
  (250000, 9)
  (350000, 13.1)
  (400000, 13.9)
  (450000, 17.1)
  (500000, 18)
  (600000, 29.2)
  (700000, 28.2)
  (800000, 33.6)
  (900000, 43.8)
  (1000000, 48)
};

\addplot[smooth,mark=x]
plot coordinates {
  (70000, 2.3)
  (90000, 2.4)
  (120000, 3.5)
  (150000, 3.8)
  (200000, 5.6)
  (250000, 6.6)
  (350000, 10.6)
  (400000, 12.3)
  (450000, 12.6)
  (500000, 15.3)
  (600000, 18.6)
  (700000, 22.3)
  (800000, 23.6)
  (900000, 28.7)
  (1000000, 31.6)
};

\end{axis}

\end{tikzpicture}

\\

\begin{tikzpicture}[scale=0.67]

\begin{axis}[xlabel=\large K \mbox{ } \mbox{ } \#proc,ylabel=t/ms,height=5cm,width=5cm]

\addplot[smooth]
plot coordinates {
  (1, 14.2)
  (2, 7.4)
  (4, 4.2)
  (6, 3.8)
  (8, 3.2)
};

\addplot[smooth,mark=o]
plot coordinates {
  (1, 27.7)
  (2, 14.1)
  (4, 8.2)
  (6, 8.5)
  (8, 7.4)
};

\addplot[smooth,mark=x]
plot coordinates {
  (1, 26.7)
  (2, 13.7)
  (4, 7.7)
  (6, 5.5)
  (8, 5.6)
};

\end{axis}

\end{tikzpicture}

&

\begin{tikzpicture}[scale=0.67]

\begin{axis}[xlabel=\large L \mbox{ } \mbox{ } \#proc,height=5cm,width=5cm]

\addplot[smooth]
plot coordinates {
  (1, 96)
  (2, 50.5)
  (4, 27.1)
  (6, 25.2)
  (8, 20.2)
};

\addplot[smooth,mark=o]
plot coordinates {
  (1, 171)
  (2, 113.1)
  (4, 54.3)
  (6, 53.5)
  (8, 44.3)
};

\addplot[smooth,mark=x]
plot coordinates {
  (1, 167.2)
  (2, 84.1)
  (4, 44.4)
  (6, 32.6)
  (8, 31.9)
};

\end{axis}

\end{tikzpicture}

\end{tabular}

\end{center}

  \caption{Quad-core i7 microbenchmarks -- $ConcurrentHashMap(-)$, $ConcurrentSkipList(\circ)$, $Ctrie(\times)$: A) $insert$, P=1; B) $insert$, P=8; C) $insert$, N=200k; D) $insert$, N=1M; E) $remove$, P=1; F) $remove$, P=8; G) $remove$, N=200k, H) $remove$, N=1M; I) $lookup$, P=1; J) $lookup$, P=8; K) $lookup$, N=200k; L) $lookup$, N=1M }
  \label{f-4xi7}
\end{figure}

\begin{figure}[h]

\begin{center}

\begin{tabular}{ c c }

\begin{tikzpicture}[scale=0.67]

\begin{axis}[
xlabel=\large A \mbox{ } \mbox{ } \#proc,ylabel=t/ms,height=5cm,width=5cm
]

\addplot[smooth]
plot coordinates {
  (1, 443)
  (2, 300)
  (4, 225)
  (6, 205)
  (8, 191)
};

\addplot[smooth,mark=x]
plot coordinates {
  (1, 558)
  (2, 289)
  (4, 156)
  (6, 123)
  (8, 93)
};

\end{axis}

\end{tikzpicture}

&

\begin{tikzpicture}[scale=0.67]

\begin{axis}[
xlabel=\large B \mbox{ } \mbox{ } \#proc,height=5cm,width=5cm
]

\addplot[smooth]
plot coordinates {
  (1, 843)
  (2, 541)
  (4, 347)
  (6, 313)
  (8, 266)
};

\addplot[smooth,mark=x]
plot coordinates {
  (1, 1273)
  (2, 635)
  (4, 337)
  (6, 285)
  (8, 227)
};

\end{axis}

\end{tikzpicture}

\\

\begin{tikzpicture}[scale=0.67]

\begin{axis}[
xlabel=\large C \mbox{ } \mbox{ } \#proc,ylabel=t/ms,height=5cm,width=5cm
]

\addplot[smooth]
plot coordinates {
  (1, 1357)
  (2, 860)
  (4, 515)
  (6, 440)
  (8, 378)
};

\addplot[smooth,mark=x]
plot coordinates {
  (1, 2164)
  (2, 1145)
  (4, 567)
  (6, 480)
  (8, 387)
};

\end{axis}

\end{tikzpicture}

&

\begin{tikzpicture}[scale=0.67]

\begin{axis}[
xlabel=\large D \mbox{ } \mbox{ } \#proc,height=5cm,width=5cm
]

\addplot[smooth]
plot coordinates {
  (1, 379)
  (2, 221)
  (4, 118)
  (6, 97)
  (8, 80)
};

\addplot[smooth,mark=x]
plot coordinates {
  (1, 564)
  (2, 290)
  (4, 156)
  (6, 122)
  (8, 92)
};

\end{axis}

\end{tikzpicture}

\\

\end{tabular}


\end{center}

  \caption{Quad-core i7 microbenchmarks -- A) $insert/lookup$, ratio=1/2, N=1M; B) $insert/lookup$, ratio=1/5, N=1M; C) $insert/lookup$, ratio=1/9, N=1M; D) $insert/lookup$ with preallocated tables, ratio=1/2, N=1M}
  \label{f-4xi7-insert-lookup}
\end{figure}
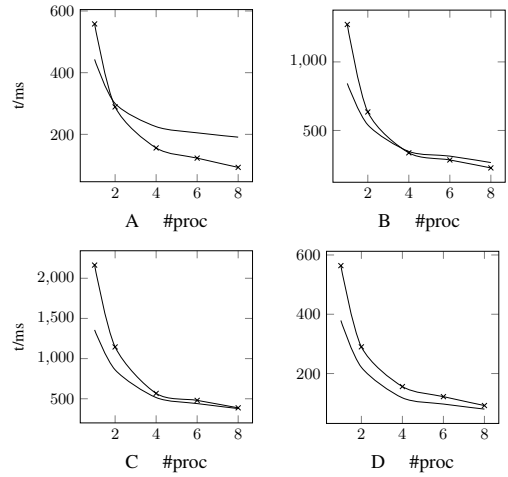

\section{Related work}

Concurrent programming techniques and important results in the area are covered by Shavit and Herlihy \cite{artofmulti08}. An overview of concurrent data structures is given by Moir and Shavit \cite{concdatastructures}. There is a body of research available focusing on concurrent lists, queues and concurrent priority queues \cite{conchashtable-harris} \cite{concpriorqueues} \cite{concqueues}. While linked lists are inefficient as sets or maps because they do not scale well, the latter two do not support the basic operations on sets and maps, so we exclude these from the further discussion and focus on more suitable data structures.

Hash tables are typically resizeable arrays of buckets. Each bucket holds some number of elements which is expected to be constant. The constant number of elements per bucket necessitates resizing the data structure. Sequential hash tables amortize the cost of resizing the table over other operations \cite{introalgs}. While the individual concurrent hash table operations such as insertion or removal can be performed in a lock-free manner as shown by Maged \cite{conchashtable-maged}, resizing is typically implemented with a global lock. Although the cost of resizal is amortized against operations by one thread, this approach does not guarantee horizontal scalability. Lea developed an extensible hash algorithm which allows concurrent searches during the resizing phase, but not concurrent insertions and removals \cite{dougleahome}. Shalev and Shavit propose split-ordered lists which keep a table of hints into a linked list in a way that does not require rearranging the elements of the linked list when resizing \cite{splitordlist}. This approach is quite innovative, but it is unclear how to shrink the hint table if most of the keys are removed, while preserving lock-freedom.

Skip lists are a data structure which stores elements in a linked list. There are multiple levels of linked lists which allow efficient insertions, removals and lookups. Skip lists were originally invented by Pugh \cite{pugh90}. Pugh proposed concurrent skip lists which achieve synchronization through the use of locks \cite{pugh90conc}. Concurrent non-blocking skip lists were later implemented by Lev, Herlihy, Luchangco and Shavit \cite{concskiplist} and Lea \cite{dougleahome}.

Concurrent binary search trees were proposed by Kung and Lehman \cite{kunglehman} -- their implementation uses a constant number of locks at a time which exclude other insertion and removal operations, while lookups can proceed concurrently. Bronson et al. presented a scalable concurrent implementation of an AVL tree based on transactional memory mechanisms which require a fixed number of locks to perform deletions \cite{bronsonavl}. Recently, the first non-blocking implementation of a binary search tree was proposed \cite{nonblockingtrees}.

Tries were originally proposed by Brandais \cite{brandaistries} and Fredkin \cite{fredkintries}. Trie hashing was applied to accessing files stored on the disk by Litwin \cite{litwinhashtrie}. Litwin, Sagiv and Vidyasankar implemented trie hashing in a concurrent setting \cite{litwinconctrie}, however, they did so by using mutual exclusion locks. Hash array mapped trees, or hash tries, are tries for shared-memory proposed by Bagwell \cite{idealhash}.
To our knowledge, there is no nonblocking concurrent implementation of hash tries prior to our work.

\section{Conclusion}

We described a lock-free concurrent implementation of the hash trie data structure. Our implementation supports insertion, remove and lookup operations. It is space-efficient in the sense that it keeps a minimal amount of information in the internal nodes. It is compact in the sense that after all removal operations complete, all paths from the root to a leaf containing a key are as short as possible. Operations are worst-case logarithmic with a low constant factor ($O(\log_{32} n)$). Its performance is comparable to that of the similar concurrent data structures. The data structure grows dynamically -- it uses no locks and there is no resizing phase. We proved that it is linearizable and lock-free.

In the future we plan to extend the algorithm with operations like \textit{move key}, which reassigns a value from one key to another atomically. One research direction is supporting efficient aggregation operations on the keys and/or stored in the Ctrie. One such specific aggregation is the size of the Ctrie -- an operation which might be useful indeed. The notion of having a size kept in one place in the Ctrie might, however, prove detrimental to the idea of distributing Ctrie operations throughout different parts of it in order to avoid contention.

Finally, we plan to develop an efficient lock-free snapshot operation for the concurrent trie which allows traversal of all the keys present in the data structure at the time at which the snapshot was created. One possible approach to doing so is to, roughly speaking, keep a partial history in the indirection nodes. A snapshot would allow traversing (in parallel) the elements of the Ctrie present at one point in time and modifying it during traversal in a way that the changes are visible only once the traversal ends. This might prove an efficient abstraction to express many iterative-style algorithms.

\bibliographystyle{splncs}


\clearpage

\pagebreak

\appendix

\setcounter{lemma}{0}
\setcounter{theorem}{0}
\setcounter{corollary}{0}
\setcounter{definition}{0}

\section{Proof of correctness}

\begin{definition}[Basics]
Value $W$ is called the \textbf{branching width}.
An \textbf{inode} $in$ is a node holding a reference $main$ to other nodes.
A \textbf{cnode} $cn$ is a node holding a bitmap $bmp$ and an set of references to other nodes called $array$.
A cnode is \textbf{$k$-way} if $length(cn.array) = k$.
An \textbf{snode} $sn$ is a node holding a key $k$ and a value $v$. An snode can be \textbf{tombed}, denoted by $sn\dagger$, meaning its tomb flag is set.
A reference $cn.arr(r)$ in the $array$ defined as $array(\#(((1 << r) - 1) \odot cn.bmp))$, where $\#( \cdot )$ is the bitcount and $\odot$ is the bitwise AND operation.
Any node $n_{l,p}$ is at \textbf{level} $l$ if there are $l / W$ cnodes on the simple path between itself and the root inode.
\textbf{Hashcode chunk} of a key $k$ at level $l$ is defined as $m(l, k) = (hashcode(k) >> l) \mod 2^W$.
A node at level $0$ has a \textbf{hashcode prefix} $p = \epsilon$, where $\epsilon$ is an empty string.
A node $n$ at level $l+W$ has a hashcode prefix $p = q \cdot r$ if and only if it can be reached from the closest parent cnode $cn_{l,q}$ by following the reference $cn_{l,q}.arr(r)$.
A reference $cn_{l,p}.sub(k)$ is defined as:

\small
\begin{align*}
cn_{l,p}.sub(k) &= 
\left\{
  \begin{array}{ll}
    cn_{l,p}.arr(m(l, k)) & \mbox{if } cn_{l,p}.flg(m(l, k)) \\
    null & \mbox{otherwise} \\
  \end{array}
\right. \\
cn_{l,p}.flg(r) &\Leftrightarrow cn_{l,p}.bmp \odot (1 \ll r) \neq 0
\end{align*}
\normalsize

\end{definition}

\begin{definition}[Ctrie]
A \textbf{Ctrie} is defined as the reference $root$ to the root of the trie.
A Ctrie \textbf{state} $\mathbb{S}$ is defined as the configuration of nodes reachable from the $root$ by following references in the nodes.
A key is within the configuration if it is in a node reachable from the root.
More formally, the relation $hasKey(in_{l,p}, k)$ on an inode $in$ at the level $l$ with a prefix $p$ and a key $k$ holds if and only if (several relations for readability):

\small
\begin{align*}
&holds(in_{l,p}, k) \Leftrightarrow in_{l,p}.main = sn: SNode \wedge sn.k = k \\
&holds(cn_{l,p}, k) \Leftrightarrow cn_{l,p}.sub(k) = sn: SNode \wedge sn.k = k \\
&hasKey(cn_{l,p}, k) \Leftrightarrow holds(cn_{l,p}, k) \vee \\
&(cn_{l,p}.sub(k)=in_{l+w,p \cdot m(l,k)} \wedge hasKey(in_{l+w,p \cdot m(l,k)}, k)) \\
&hasKey(in_{l,p}, k) \Leftrightarrow holds(in_{l,p}, k) \vee \\
&(in_{l,p}.main = cn_{l,p}: CNode \wedge hasKey(cn_{l,p}, k)) \\
\end{align*}
\normalsize
\end{definition}

\begin{definition}
We define the following invariants for the Ctrie.
\small
\begin{description}
\item[INV1] $inode_{l,p}.main = null | cnode_{l,p} | snode\dagger$
\item[INV2] $\#(cn.bmp) = length(cn.array)$
\item[INV3] $cn_{l,p}.flg(r) \neq 0 \Leftrightarrow cn_{l,p}.arr(r) \in \{sn, in_{l+W,p \cdot r}\}$
\item[INV4] $cn_{l,p}.arr(r) = sn \Leftrightarrow hashcode(sn.k) = p \cdot r \cdot s$
\item[INV5] $in_{l,p}.main = sn\dagger \Leftrightarrow hashcode(sn.k) = p \cdot r$
\end{description}
\normalsize
\end{definition}

\begin{definition}[Validity]
A state $\mathbb{S}$ is \textbf{valid} if and only if the invariants INV1-5 are true in the state $\mathbb{S}$.
\end{definition}

\begin{definition}[Abstract set]
An \textbf{abstract set} $\mathbb{A}$ is a mapping $K \Rightarrow \{\bot, \top\}$ which is true only for those keys which are a part of the abstract set, where $K$ is the set of all keys. An \textbf{empty abstract set} $\varnothing$ is a mapping such that $\forall k, \varnothing(k) = \bot$. Abstract set operations are $insert(k, \mathbb{A}) = \mathbb{A}_1: \forall k'\in\mathbb{A}_1, k'=k \lor k'\in\mathbb{A}$, $lookup(k, \mathbb{A}) = \top \Leftrightarrow k\in\mathbb{A}$ and $remove(k, \mathbb{A}) = \mathbb{A}_1: k \not\in \mathbb{A}_1 \wedge \forall k'\in\mathbb{A}, k \neq k' \Rightarrow k'\in\mathbb{A}$. Operations $insert$ and $remove$ are \textbf{destructive}.
\end{definition}

\begin{definition}[Consistency]
A Ctrie state $\mathbb{S}$ is \textbf{consistent} with an abstract set $\mathbb{A}$ if and only if $k\in\mathbb{A} \Leftrightarrow hasKey(Ctrie, k)$.
A destructive Ctrie operation $op$ is \textbf{consistent} with the corresponding abstract set operation $op'$ if and only if applying $op$ to a state $\mathbb{S}$ consistent with $\mathbb{A}$ changes the state into $\mathbb{S}'$ consistent with an abstract set $\mathbb{A}' = op(k, \mathbb{A})$.
A Ctrie $lookup$ is \textbf{consistent} with the abstract set lookup if and only if it returns the same value as the abstract set $lookup$, given that the state $\mathbb{S}$ is consistent with $\mathbb{A}$.
A \textbf{consistency change} is a change from state $\mathbb{S}$ to state $\mathbb{S}'$ of the Ctrie such that $\mathbb{S}$ is consistent with an abstract set $\mathbb{A}$ and $\mathbb{S}'$ is consistent with an abstract set $\mathbb{A}'$ and $\mathbb{A} \neq \mathbb{A}'$.
\end{definition}

We point out that there are multiple valid states corresponding to the same abstract set.

\begin{theorem}[Safety]
At all times $t$, a Ctrie is in a valid state $\mathbb{S}$, consistent with some abstract set $\mathbb{A}$. All Ctrie operations are consistent with the semantics of the abstract set $\mathbb{A}$.
\end{theorem}

First, it is trivial to see that if the state $\mathbb{S}$ is valid, then the Ctrie is also consistent with some abstract set $\mathbb{A}$. Second, we prove the theorem using structural induction. As induction base, we take the empty Ctrie which is valid and consistent by definition. The induction hypothesis is that the Ctrie is valid and consistent at some time $t$. We use the hypothesis implicitly from this point on. Before proving the induction step, we introduce additional definitions and lemmas.

\begin{definition}
A node is \textbf{live} if and only if it is a cnode, a non-tombed snode or an inode whose $main$ reference points to a cnode.
A \textbf{nonlive} node is a node which is not live.
A \textbf{null-inode} is an inode with a $main$ set to $null$.
A \textbf{tomb-inode} is an inode with a $main$ set to a tombed snode $sn\dagger$.
A node is a \textbf{singleton} if it is an snode or an inode $in$ such that $in.main = sn\dagger$, where $sn\dagger$ is tombed.
\end{definition}

\begin{lemma}[End of life]\label{l-nonlive}
If an inode $in$ is either a null-inode or a tomb-inode at some time $t_0$, then $\forall t > t_0$ $in.main$ is not written.
\end{lemma}

\begin{proof}
For any inode $in$ which becomes reachable in the Ctrie at some time $t$, all assignments to $in.main$ at any time $t_0 > t$ occur in a CAS instruction -- we only have to inspect these writes.

Every CAS instruction on $in.main$ is preceeded by a check that the expected value of $in.main$ is a cnode. From the properties of CAS, it follows that if the current value is either $null$ or a tombed snode, the CAS will not succeed. Therefore, neither null-inodes nor tomb-inodes can be written to $in.main$.
\end{proof}

\begin{lemma}\label{l-immut}
Cnodes and snodes are immutable -- once created, they no longer change the value of their fields.
\end{lemma}

\begin{proof}
Trivial inspection of the pseudocode reveals that $k$, $v$, $tomb$, $bmp$ and $array$ are never assigned a value after an snode or a cnode was created.
\end{proof}

\begin{definition}
A \textbf{compression} $ccn$ of a cnode $cn$ seen at some time $t_0$ is a node such that:
\begin{itemize}
\item $ccn = sn\dagger$ if $length(cn.array) = 1$ and $cn.array(0).main = sn\dagger$ at $t_0$
\item $ccn = null$ if $\forall i, cn.array(i).main = null$ at $t_0$ (including the case where $length(cn.array) = 0$)
\item otherwise, $ccn$ is a cnode obtained from $cn$ so that at least those null-inodes existing at $t_0$ are removed and at least those tomb-inodes $in$ existing at $t_0$ are resurrected - that is, replaced by untombed copies $sn$ of $sn\dagger = in.main$
\end{itemize}

A \textbf{weak tombing} $wtc$ of a cnode $cn$ seen at some time $t_0$ is a node such that:
\begin{itemize}
\item $ccn = sn\dagger$ if $length(cn.array) = 1$ and $cn.array(0)$ is a tomb-inode or an snode at $t_0$
\item $ccn = null$ if $\forall i, cn.array(i).main = null$ at $t_0$
\item $ccn = cn$ if there is more than a single non-null-inode below at $t_0$
\item otherwise, $ccn$ is a one-way cnode obtained from $cn$ such that all null-inodes existing at $t_0$ are removed
\end{itemize}
\end{definition}

\begin{lemma}\label{l-comptomb}
Methods $toCompressed$ and $toWeakTombed$ return the compression and weak tombing of a cnode $cn$, respectively.
\end{lemma}

\begin{proof}
From lemma \ref{l-immut} we know that a cnode does not change values of $bmp$ or $array$ once created. From lemma \ref{l-nonlive} we know that all the nodes that are nonlive at $t_0$ must be nonlive $\forall t > t_0$. Methods $toCompressed$ or $toWeakTombed$ scan the array of $cn$ sequentially and make checks which are guaranteed to stay true if they pass -- when these methods complete at some time $t > t_0$ they will have removed or resurrected at least those inodes that were nonlive at some point $t_0$ after the operation began.
\end{proof}

\begin{lemma}\label{l-inv123}
Invariants INV1, INV2 and INV3 are always preserved.
\end{lemma}

\begin{proof}
INV1: Inode creation and every CAS instruction abide this invariant. There are no other writes to $main$.

INV2, INV3: Trivial inspection of the pseudocode shows that the creation of cnodes abides these invariants. From lemma \ref{l-immut} we know that cnodes are immutable. Therefore, these invariants are ensured during construction and do not change subsequently.
\end{proof}

\begin{lemma}\label{l-parentreach}
If any CAS instruction makes an inode $in$ unreachable from its parent at some time $t$, then $in$ is nonlive at time $t$.
\end{lemma}

\begin{proof}
We will show that all the inodes a CAS instruction could have made unreachable from their parents at some point $t_1$ were nonlive at some time $t_0 < t_1$. The proof then follows directly from lemma \ref{l-nonlive}. We now analyze successful CAS instructions.

In lines \ref{cas_topinsert}, \ref{cas_topremove} and \ref{cas_toplookup}, if $r$ is an inode and it is removed from the trie, then it must have been previously checked to be a null-inode in lines \ref{check_topinsert}, \ref{check_topremove} and \ref{check_toplookup}, respectively.

In lines \ref{cas_insertcn}, \ref{cas_insertsn} and \ref{cas_insertin}, a cnode $cn$ is replaced with a new cnode $ncn$ which contains all the references to inodes as $cn$ does, and possibly some more. These instructions do not make any inodes unreachable.

In line \ref{cas_remove}, a cnode $cn$ is replaced with a new $ncn$ which contains all the node references as $cn$ but without one reference to an snode -- all the inodes remain reachable.

In line \ref{cas_clean}, a cnode $m$ is replaced with its compression $mc$ -- from lemma \ref{l-comptomb}, $mc$ may only be deprived of references to nonlive inodes.

In line \ref{cas_tomb}, a cnode $m$ is replaced with its weak tombing $mwt$ -- from lemma \ref{l-comptomb}, $mwt$ may only be deprived of references to nonlive inodes.
\end{proof}

\begin{corollary}\label{c-reach}
Lemma \ref{l-parentreach} has a consequence that any inode $in$ can only be made unreachable in the Ctrie through modifications in their parent inode (or the root reference if $in$ is referred by it). If there is a parent that refers to $in$, then that parent is live by definition. If the parent had been previously removed, lemma \ref{l-parentreach} tells us that the parent would have been nonlive at the time. From lemma \ref{l-nonlive} we know that the parent would remain nonlive afterwards. This is a contradiction.
\end{corollary}

\begin{lemma}\label{l-laterreach}
If at some time $t_1$ an inode $in$ is read by some thread (lines \ref{read_topinsert}, \ref{read_topremove}, \ref{read_toplookup}, \ref{read_inlookup}, \ref{read_ininsert}, \ref{read_inremove}, \ref{read_incontract}), followed by a read of cnode $cn = in.main$ in the same thread at time $t_2 > t_1$ (lines \ref{read_lookup}, \ref{read_insert}, \ref{read_remove}, \ref{read_clean}, \ref{read_tomb}, \ref{read_contract}), then $in$ is reachable from the root at time $t_2$. Trivially, so is $in.main$.
\end{lemma}

\begin{proof}
Assume, that inode $in$ is not reachable from the root at $t_2$. That would mean that $in$ was made unreachable at an earlier time $t_0 < t_2$. Corollary \ref{c-reach} says that $in$ was then nonlive at $t_0$. However, from lemma \ref{l-nonlive} it follows that $in$ must be nonlive for all times greater than $t_0$, including $t_2$. This is a contradiction -- $in$ is live at $t_2$, since it contains a cnode $cn = in.main$.
\end{proof}

\begin{lemma}[Presence]\label{l-presence}
Reading a cnode $cn$ at some time $t_0$ and then $cn.sub(k)$ such that $k = sn.k$ at some time $t_1 > t_0$ means that the relation $hasKey(root, k)$ holds at time $t_0$. Trivially, $k$ is then in the corresponding abstract set $\mathbb{A}$.
\end{lemma}

\begin{proof}
We know from lemma \ref{l-laterreach} that the corresponding cnode $cn$ was reachable at some time $t_0$. Lemma \ref{l-immut} tells us that $cn$ and $sn$ were the same $\forall t > t_0$. Therefore, $sn$ was present in the array of $cn$ at $t_0$, so it was reachable. Furthermore, $sn.k$ is the same $\forall t > t_0$. It follows that $hasKey(root, x)$ holds at time $t_0$.
\end{proof}

\begin{definition}
A \textbf{longest path} of nodes $\pi(h)$ for some hashcode $h$ is the sequence of nodes from the root to a leaf of a valid Ctrie such that:

\begin{itemize}
\item if $root = null$ then $\pi(h) = \epsilon$
\item if $root \neq null$ then the first node in $\pi(h)$ is $root$, which is an inode
\item $\forall in \in \pi(h)$ if $in.main = cn$, then the next element in the path is $cn$
\item $\forall in \in \pi(h)$ if $in.main = sn$, then the last element in the path is $sn$
\item $\forall in \in \pi(h)$ if $in.main = null$, then the last element in the path is $in$
\item $\forall cn_{l,p} \in \pi(h), h = p \cdot r \cdot s$ if $cn.flg(r) = \bot$, then the last element in the path is $cn$, otherwise the next element in the path is $cn.arr(r)$
\end{itemize}
\end{definition}

\begin{lemma}[Longest path]\label{l-lpath}
Assume that a non-empty Ctrie is in a valid state at some time $t$. The longest path of nodes $\pi(h)$ for some hashcode $h = r_0 \cdot r_1 \cdots r_n$ is a sequence $in_{0,\epsilon} \rightarrow cn_{0,\epsilon} \rightarrow in_{W,r_0} \rightarrow \ldots \rightarrow in_{W \cdot m, r_0 \cdots r_m} \rightarrow x$, where $x \in \{cn_{W \cdot m, r_0 \cdots r_m}, sn, cn_{W \cdot m, r_0 \cdots r_m} \rightarrow sn, null\}$.
\end{lemma}

\begin{proof}
Trivially from the invariants and the definition of the longest path.
\end{proof}



\begin{lemma}[Absence I]\label{l-absence1}
Assume that at some time $t_0$ $\exists cn = in.main$ for some node $in_{l,p}$ and the algorithm is searching for a key $k$. Reading a cnode $cn$ at some time $t_0$ such that $cn.sub(k) = null$ and $hashcode(k) = p \cdot r \cdot s$ implies that the relation $hasKey(root, k)$ does not hold at time $t_0$. Trivially, $k$ is not in the corresponding abstract set $\mathbb{A}$.
\end{lemma}

\begin{proof}
Lemma \ref{l-laterreach} implies that $in$ is in the configuration at time $t_0$, because $cn = cn_{l,p}$ such that $hashcode(k) = p \cdot r \cdot s$ is live. The induction hypothesis states that the Ctrie was valid at $t_0$. We prove that $hasKey(root, k)$ does not hold by contradiction. Assume there exists an snode $sn$ such that $sn.k = k$. By lemma \ref{l-lpath}, $sn$ can only be the last node of the longest path $\pi(h)$, and we know that $cn$ is the last node in $\pi(h)$.


\end{proof}

\begin{lemma}[Absence II]\label{l-absence2}
Assume that the algorithm is searching for a key $k$. Reading a live snode $sn$ at some time $t_0$ and then $x = sn.k \neq k$ at some time $t_1 > t_0$ means that the relation $hasKey(root, x)$ does not hold at time $t_0$. Trivially, $k$ is not in the corresponding abstract set $\mathbb{A}$.
\end{lemma}

\begin{proof}
Contradiction similar to the one in the previous lemma.
\end{proof}

\begin{lemma}[Absence III]\label{l-absence3}
Assume that the $root$ reference is read in $r$ at $t_0$ and $r$ is positively compared to $null$ at $t_1 > t_0$. Then $\forall k, hasKey(root, k)$ does not hold at $t_0$. Trivially, the Ctrie is consistent with the empty abstract set $\varnothing$.
\end{lemma}

\begin{proof}
Local variable $r$ has the same value $\forall t \ge t_0$. Therefore, at $t_0$ $root = null$. The rest is trivial.
\end{proof}

\begin{lemma}[Fastening]\label{l-fastening}
1. Assume that one of the CAS instructions in lines \ref{cas_topinsert}, \ref{cas_topremove} and \ref{cas_toplookup} succeeds at time $t_1$ after $r$ was determined to be a nonlive inode in one of the lines \ref{check_topinsert}, \ref{check_topremove} or \ref{check_toplookup}, respectively, at time $t_0$. Then $\forall t, t_0 \le t < t_1$, relation $hasKey(root, k)$ does not hold for any key. If $r$ is $null$, then $\exists \delta > 0 \forall t, t_1 - \delta < t < t_1$ $hasKey$ does not hold for any key.

2. Assume that one of the CAS instructions in lines \ref{cas_insertcn} and \ref{cas_insertin} succeeds at time $t_1$ after $in.main$ was read in line \ref{read_insert} at time $t_0$. The $\forall t, t_0 \le t < t_1$, relation $hasKey(root, k)$ does not hold.

3. Assume that the CAS instruction in line \ref{cas_insertsn} succeeds at time $t_1$ after $in.main$ was read in line \ref{read_insert} at time $t_0$. The $\forall t, t_0 \le t < t_1$, relation $hasKey(root, k)$ holds.

4. Assume that the CAS instruction in line \ref{cas_remove} succeeds at time $t_1$ after $in.main$ was read in line \ref{read_remove} at time $t_0$. The $\forall t, t_0 \le t < t_1$, relation $hasKey(root, k)$ holds.
\end{lemma}

\begin{proof}
The algorithm never creates a reference to a newly allocated memory areas unless that memory area has been previously reclaimed. Although it is possible to extend the pseudocode with memory management directives, we omit memory-reclamation from the pseudocode and assume the presence of a garbage collector which does not reclaim memory areas as long as there are references to them reachable from the program. In the pseudocode, CAS instructions always work on memory locations holding references -- $CAS(x, r, r')$ takes a reference $r$ to a memory area allocated for nodes as its expected value, meaning that a reference $r$ that is reachable in the program exists from the time $t_0$ when it was read until $CAS(x, r, r')$ was invoked at $t_1$. On the other hand, the new value for the CAS is in all cases a newly allocated object. In the presence of a garbage collector with the specified properties, a new object cannot be allocated in any of the areas still being referred to. It follows that if a CAS succeeds at time $t_1$, then $\forall t, t_0 \le t < t_1$, where $t_0$ is the time of reading a reference and $t_1$ is the time when CAS occurs, the corresponding memory location $x$ had the same value $r$.

We now analyze specific cases from the lemma statement:

1. We know that $\forall t, t_0 \le t < t_1$ the root reference has a reference $r$ to the same inode $in$. We assumed that $r$ is nonlive at $t_0$. From lemma \ref{l-nonlive} it follows that $r$ remains nonlive until time $t_1$. By the definition of $hasKey$, the relation does not hold for any key from $\forall t, t_0 \le t < t_1$. Case where $r$ is $null$ is proved similarly.

2. From lemma \ref{l-lpath} we know that for some hashcode $h = hashcode(k)$ there exists a longest path of nodes $\pi(h) = in_{0,\epsilon} \rightarrow \ldots \rightarrow cn_{l,p}$ such that $h = p \cdot r \cdot s$ and that $sn$ such that $sn.k = k$ cannot be a part of this path -- it could only be referenced by $cn_{l,p}.sub(k)$ of the last cnode in the path.
We know that $\forall t, t_0 \le t < t_1$ reference $cn$ points to the same cnode. We know from \ref{l-immut} that cnodes are immutable. The check to $cn.bmp$ preceeding the CAS ensures that $\forall t, t_0 \le t < t_1$ $cn.sub(k) = null$. In the other case, we check that the key $k$ is not contained in $sn$. We know from \ref{l-parentreach} that $cn$ is reachable during this time, because $in$ is reachable.
Therefore, $hasKey(root, k)$ does not hold $\forall t, t_0 \le t < t_1$.

3., 4. We know that $\forall t, t_0 \le t < t_1$ reference $cn$ points to the same cnode. Cnode $cn$ is reachable as long as its parent inode $in$ is reachable. We know that $in$ is reachable by lemma \ref{l-parentreach}, since $in$ is live $\forall t, t_0 \le t < t_1$. We know that $cn$ is immutable by lemma \ref{l-immut} and that it contains a reference to $sn$ such that $sn.k = k$. Therefore, $sn$ is reachable and $hasKey(root, k)$ holds $\forall t, t_0 \le t < t_1$.
\end{proof}

\begin{lemma}\label{l-consmodif}
Assume that the Ctrie is valid and consistent with some abstract set $\mathbb{A}$ $\forall t, t_1 - \delta < t < t_1$. CAS instructions from lemma \ref{l-fastening} induce a change into a valid state which is consistent with the abstract set semantics.
\end{lemma}

\begin{proof}
From lemma \ref{l-fastening}, we know that a successful CAS in line \ref{cas_topinsert} means that the Ctrie was consistent with an empty abstract set $\varnothing$ up to some time $t_1$. After that time, the Ctrie is consistent with the abstract set $\mathbb{A} = {k}$. Successful CAS instructions in lines \ref{cas_topremove} and \ref{cas_toplookup} mean that the Ctrie was consistent with an empty abstract set $\varnothing$ up to some time $t_1$ and are also consistent with $\varnothing$ at $t_1$.

Observe a successful CAS in line \ref{cas_insertcn} at some time $t_1$ after $cn$ was read in line \ref{read_insert} at time $t_0 < t_1$. From lemma \ref{l-fastening} we know that $\forall t, t_0 \le t < t_1$, relation $hasKey(root, k)$ does not hold. If the last CAS instruction in the Ctrie occuring before the CAS in line \ref{cas_clean} was at $t_\delta = t_1 - \delta$, then we know that $\forall t, \max(t_0, t_\delta) \le t < t_1$ the $hasKey$ relation does not change. We know that at $t_1$ $cn$ is replaced with a new cnode with a reference to a new snode $sn$ such that $sn.k = k$, so at $t_1$ relation $hasKey(root, k)$ holds. Consequently, up to $\forall t, \max(t_0, t_\delta) \le t < t_1$ the Ctrie is consistent with an abstract set $\mathbb{A}$ and at $t_1$ it is consistent with an abstract set $\mathbb{A} \cup \{ k \}$.
Validity is trivial.

Proofs for the CAS instructions in lines \ref{cas_insertsn}, \ref{cas_insertin} and \ref{cas_remove} are similar.
\end{proof}

\begin{lemma}\label{l-consclean}
Assume that the Ctrie is valid and consistent with some abstract set $\mathbb{A}$ $\forall t, t_1 - \delta < t < t_1$. If one of the operations $clean$, $tombCompress$ or $contractParent$ succeeds with a CAS at $t_1$, the Ctrie will remain valid and consistent with the abstract set $\mathbb{A}$ at $t_1$.
\end{lemma}

\begin{proof}
Operations $clean$, $tombCompress$ and $contractParent$ are atomic - their linearization point is the first successful CAS instruction occuring at $t_1$. We know from lemma \ref{l-comptomb} that methods $toCompressed$ and $toWeakTombed$ produce a compression and a weak tombing of a cnode, respectively.

We first prove the property $\exists k, hasKey(cn, k) \Rightarrow hasKey(f(cn), k)$, where $f$ is either a compression or a weak tombing. We know from their respective definitions that the resulting cnode $ncn = f(cn)$ or the result $null = f(cn)$ may only omit nonlive inodes from $cn$. Omitting a null-inode omits no key. Omitting a tomb-inode may omit exactly one key, but that is compensated by adding new snodes -- $sn\dagger$ in the case of a one-way node or, with compression, resurrected copies $sn$ of removed inodes $in$ such that $in.main = sn\dagger$. Therefore, the $hasKey$ relation is exactly the same for both $cn$ and $f(cn)$.

We only have to look at cases where CAS instructions succeed. If CAS in line \ref{cas_clean} at time $t_1$ succeeds, then $\forall t, t_0 < t < t_1$ $in.main = cn$ and at $t_1$ $in.main = toCompressed(cn)$. Assume there is some time $t_\delta = t_1 - \delta$ at which the last CAS instruction in the Ctrie occuring before the CAS in line \ref{cas_clean} occurs. Then $\forall t, \max(t_0, t_\delta) \le t < t_1$ the $hasKey$ relation does not change. Additionally, it does not change at $t_1$, as shown above. Therefore, the Ctrie remains consistent with the abstract set $\mathbb{A}$.
Validity is trivial.

Proof for $tombCompress$ and $contractParent$ is similar.
\end{proof}

\begin{corollary}\label{c-inv45}
From lemmas \ref{l-consmodif} and \ref{l-consclean} it follows that invariants INV4 and INV5 are always preserved.
\end{corollary}

\begin{proof}[Safety]
We proved at this point that the algorithm is safe - Ctrie is always in a valid (lemma \ref{l-inv123} and corollary \ref{c-inv45}) state consistent with some abstract set. All operations are consistent with the abstract set semantics (lemmas \ref{l-presence}, \ref{l-absence1}, \ref{l-absence2}, \ref{l-absence3} \ref{l-consmodif} and \ref{l-consclean}).
\end{proof}

\begin{theorem}[Linearizability]
Operations $insert$, $lookup$ and $remove$ are linearizable.
\end{theorem}

\begin{proof}[Linearizability]
An operation is linearizable if we can identify its linearization point. The linearization point is a single point in time when the consistency of the Ctrie changes. The CAS instruction itself is linearizable, as well as atomic reads. It is known that a single invocation of a linearizable instruction has a linearization point.

1. We know from lemma \ref{l-consclean} that operation $clean$ does not change the state of the corresponding abstract set. Operation $clean$ is followed by a restart of the operation it was called from and is not preceeded by a consistency change -- all successful writes in the $insert$ and $iinsert$ that change the consistency of the Ctrie result in termination.

CAS in line \ref{cas_topinsert} that succeeds at $t_1$ immediately returns. By lemma \ref{l-consmodif}, $\exists \delta > 0 \forall t, t_1 - \delta < t < t_1$ the Ctrie is consistent with an empty abstract set $\varnothing$, and at $t_1$ it is consistent with $\mathbb{A} = \{k\}$. If this is the first invocation of $insert$, then the CAS is the first and the last write with consistent semantics. If $insert$ has been recursively called, then it has not been preceeded by a consistency change -- no successful CAS instruction in $iinsert$ is followed by a recursive call to the method $insert$. Therefore, it is the linearization point.

CAS in line \ref{cas_insertcn} that succeeds at $t_1$ immediately returns. By lemma \ref{l-consmodif}, $\exists \delta > 0 \forall t, t_1 - \delta < t < t_1$ the Ctrie is consistent with an empty abstract set $\mathbb{A}$ and at $t_1$ it is consistent with $\mathbb{A} \cup \{ k \}$. If this is the first invocation of $iinsert$, then the CAS is the first and the last write with consistent semantics. If $iinsert$ has been recursively called, then it was preceeded by an $insert$ or $iinsert$. We have shown that if its preceeded by a call to $insert$, then there have been no preceeding consistency changes. If it was preceeded by $iinsert$, then there has been no write in the previous $iinsert$ invocation. Therefore, it is the linearization point.

Similar arguments hold for CAS instructions in lines \ref{cas_insertsn} and \ref{cas_insertin}. It follows that if some CAS instruction in the $insert$ invocation is successful, then it is the only successful CAS instruction. Therefore, $insert$ is linearizable.

2. Operation $clean$ is not preceeded by a write that results in a consistency change and does not change the consistency of the Ctrie.

Assume that a check in line \ref{check_toplookup} succeeds. The state of the local variable $r$ does not change $\forall t > t_0$ where $t_0$ is the atomic read in the preceeding line \ref{read_toplookup}. The linearization point is then the read at $t_0$, by lemma \ref{l-absence3}.

Assume that a CAS in line \ref{cas_toplookup} succeeds at $t_1$. By lemma \ref{l-consmodif}, $\exists \delta > 0 \forall t, t_1 - \delta < t < t_1$ the Ctrie is consistent with an empty abstract set $\varnothing$, and at $t_1$ it is consistent with $\varnothing$. Therefore, this write does not result in consistency change and is not preceeded by consistency changes. This write is followed by the restart of the operation.

Assume that a node $m$ is read in line \ref{read_lookup} at $t_0$. By lemma \ref{l-immut}, if $cn.sub(k) = null$ at $t_1$ then $\forall t, cn.sub(k) = null$. By corollary \ref{c-reach}, $cn$ is reachable at $t_0$, so at $t_0$ the relation $hasKey(root, k)$ does not hold. The read at $t_0$ is not preceeded by a consistency changing write and followed by a termination of the $lookup$ so it is a linearization point if the method returns in line \ref{flag_lookup}. By similar reasoning, if the operation returns in lines \ref{check_lookup} or \ref{notcheck_lookup}, the read in line \ref{read_lookup} is the linearization point..

We have identified linearization points for the $lookup$, therefore $lookup$ is linearizable.

3. Operation $clean$ is not preceeded by a write that results in a consistency change and does not change the consistency of the Ctrie.

By lemma \ref{l-consclean} operations $tombCompress$ and $contractParent$ do not cause a consistency change. Furthermore, they are only followed by calls to $tombCompress$ and $contractParent$ and the termination of the operation.

Assume that the check in line \ref{check_topremove} succeeds after the read in line \ref{read_topremove} at time $t_0$. By applying the same reasoning as for $lookup$ above, the read at time $t_0$ is the linearization point.

Assume CAS in line \ref{cas_topremove} succeeds at $t_1$. We apply the same reasoning as for $lookup$ above -- this instruction does not change the consistency of the Ctrie and is followed by a restart of the operation.

Assume that a node $m$ is read in line \ref{read_remove} at $t_0$. By similar reasoning as with $lookup$ above, the read in line \ref{read_remove} is a linearization point if the method returns in either of the lines \ref{flag_remove} or \ref{notcheck_remove}.

Assume that the CAS in line \ref{cas_remove} succeeds at time $t_0$. By lemma \ref{l-consmodif}, $\exists \delta > 0 \forall t, t_1 - \delta < t < t_1$ the Ctrie is consistent with an empty abstract set $\mathbb{A}$ and at $t_1$ it is consistent with $\mathbb{A} \setminus \{ k \}$. This write is not preceeded by consistency changing writes and followed only by $tombCompress$ and $contractParent$ which also do not change consistency. Therefore, it is a linearization point.

We have identified linearization points for the $remove$, therefore $remove$ is linearizable.
\end{proof}

\begin{definition}
Assume that a multiple number of threads are invoking a concurrent operation $op$. The concurrent operation $op$ is \textbf{lock-free} if and only if after a finite number of thread execution steps some thread completes the operation.
\end{definition}

\begin{theorem}[Lock-freedom]
Ctrie operations $insert$, $lookup$ and $remove$ are lock-free.
\end{theorem}

The rough idea of the proof is the following. To prove lock-freedom we will first show that there is a finite number of steps between state changes. Then we define a space of possible states and show that there can only be finitely many successful CAS instructions which do not result in a consistency change. We have shown in lemmas \ref{l-consmodif} and \ref{l-consclean} that only CAS instructions in lines \ref{cas_topremove}, \ref{cas_toplookup}, \ref{cas_tomb}, \ref{cas_contractnull} and \ref{cas_contractsingle} do not cause a consistency change. We proceed by introducing additional definitions and prooving the necessary lemmas. In all cases, we assume there has been no state change which is a consistency change, otherwise that would mean that some operation was completed.

\begin{lemma}\label{l-roottomb}
The $root$ is never a tomb-inode.
\end{lemma}

\begin{proof}
A tomb-inode can only be assigned to $in.main$ of some $in$ in $clean$ and $tombCompress$. Neither $clean$ nor $tombCompress$ are ever called for the $in$ in the root of the Ctrie, as they are preceeded by the check if $parent$ is different than $null$.
\end{proof}

\begin{lemma}\label{l-failedcas}
If a CAS that does not cause a consistency change in one of the lines \ref{cas_insertcn}, \ref{cas_insertsn}, \ref{cas_insertin}, \ref{cas_clean}, \ref{cas_tomb}, \ref{cas_contractnull} or \ref{cas_contractsingle} fails at some time $t_1$, then there has been a state change since the time $t_0$ when a respective read in one of the lines \ref{read_insert}, \ref{read_insert}, \ref{read_insert}, \ref{read_clean}, \ref{read_tomb}, \ref{read_contract} or \ref{read_contract} occured. Trivially, the state change preceeded the CAS by a finite number of execution steps.
\end{lemma}

\begin{proof}
The configuration of nodes reachable from the root has changed, since the corresponding $in.main$ has changed. Therefore, the state has changed by definition.
\end{proof}

\begin{lemma}\label{l-betweencas}
In each operation there is a finite number of execution steps between consecutive CAS instructions.
\end{lemma}

\begin{proof}
The $ilookup$ and $iinsert$ operations have a finite number of executions steps. There are no loops in the pseudocode for $ilookup$ in $iinsert$, the recursive calls to them occur on the lower level of the trie and the trie depth is bound -- no non-consistency changing CAS increases the depth of the trie.

The $lookup$ operation is restarted if and only if there has been a CAS in line \ref{cas_toplookup} or if $clean$ (which contains a CAS) is called in $ilookup$. If $clean$ was not called in $ilookup$ after the check that the parent is not $null$ at $t_0$, then $root$ was $in$ such that $in.main = null$ at $t_0$ (it is not tombed by lemma \ref{l-roottomb}). Assuming there has been no state change, the CAS will occur in the next recursive call to $lookup$.

The $insert$ operation is restarted if and only if there has been a CAS in line \ref{cas_topinsert} or if $clean$ (which contains a CAS) is called in $iinsert$. If $clean$ was not called in $iinsert$ after the check that the parent is not $null$ at $t_0$, then $root$ was $in$ such that $in.main = null$ at $t_0$. Assuming no state change, a CAS will occur in the next recursive call to $insert$.

The $insert$ operation can also be restarted due to a preceeding failed CAS in lines \ref{cas_insertcn}, \ref{cas_insertsn} or \ref{cas_insertin}. By lemma \ref{l-failedcas}, there must have been a state change in this case.

In $iremove$, calls to $tombCompress$ and $contractParent$ contain no loops, but are recursive. In case they restart themselves, a CAS is invoked at least once. Between these CAS instructions there is a finite number of execution steps.

A similar analysis as for $lookup$ above can be applied to the first phase of $remove$ which consists of all the execution steps preceeding a successful CAS in line \ref{cas_remove}. The number of times $tombCompress$ and $contractParent$ from the $iremove$ in the cleanup phase is bound by the depth of the trie and there is a finite number of execution steps between them. Once the root is reached, $remove$ completes.

Therefore, all operations have a finite number of executions steps between consecutive CAS instructions, assuming that the state has not changed since the last CAS instruction.
\end{proof}

\begin{corollary}\label{c-finexec}
The consequence of lemmas \ref{l-betweencas} and \ref{l-failedcas} is that there is a finite number of execution steps between two state changes. At any point during the execution of the operation we know that the next CAS instruction is due in a finite number of execution steps (lemma \ref{l-betweencas}). From lemmas \ref{l-consmodif} and \ref{l-consclean} we know that if a CAS succeeds, it changes the state. From lemma \ref{l-failedcas} we know that if the CAS fails, the state was changed by someone else.
\end{corollary}

We remark at this point that corollary \ref{c-finexec} does not imply that there is a finite number of execution steps between two operations. A state change is not necessarily a consistency change.

\begin{definition}
Let there at some time $t_0$ be a $1$-way cnode $cn$ such that $cn.array(0) = in$ and $in.main = sn\dagger$ where $sn\dagger$ is tombed or, alternatively, $cn$ is a $0$-way node. We call such $cn$ a \textbf{single tip of length $1$}. Let there at some time $t_0$ be a 1-way cnode $cn$ such that $cn.array(0) = cn'$ and $cn'$ is a single tip of length $k$. We call such $cn$ a \textbf{single tip of length $k + 1$}.
\end{definition}

\begin{definition}
The \textbf{total path length} $d$ is the sum of the lengths of all the paths from the root to some leaf.
\end{definition}

\begin{definition}
Assume the Ctrie is in a valid state. Let $n$ be the number of reachable null-inodes in this state, $t$ the number of reachable tomb-inodes, $l$ the number of live inodes, $r$ the number of single tips of any length and $d$ the total path length. We denote the state of the Ctrie as $\mathbb{S}_{n,t,l,r,d}$. We call the state $\mathbb{S}_{0,0,l,r,d}$ the \textbf{clean} state.
\end{definition}

\begin{lemma}\label{l-statechanges}
Observe all CAS instructions which never cause a consistency change and assume they are successful. Assuming there was no state change since reading $in$ prior to calling $clean$, the CAS in line \ref{cas_clean} changes the state of the Ctrie from the state $\mathbb{S}_{n,t,l,r,d}$ to either $\mathbb{S}_{n+j,t,l,r-1,d-1}$ where $r > 0$, $j \in \{ 0, 1 \}$ and $d \ge 1$, or to $\mathbb{S}_{n-k,t-j,l,r,d' \le d}$ where $k \ge 0$, $j \ge 0$, $k + j > 0$, $n \ge k$ and $t \ge j$.

Furthermore, the CAS in line \ref{cas_topremove} changes the state of the Ctrie from $\mathbb{S}_{1,0,0,0,1}$ to $\mathbb{S}_{0,0,0,0,0}$. The CAS in line \ref{cas_toplookup} changes the state from $\mathbb{S}_{1,0,0,0,1}$ to $\mathbb{S}_{0,0,0,0,0}$. The CAS in line \ref{cas_tomb} changes the state from $\mathbb{S}_{n,t,l,r,d}$ to either $\mathbb{S}_{n+j,t,l,r-1,d-j}$ where $r > 0$, $j \in \{ 0, 1 \}$ and $d \ge j$, or to $\mathbb{S}_{n-k,t,l,r,d' \le d}$ where $k > 0$ and $n \ge k$. The CAS in line \ref{cas_contractnull} changes the state from $\mathbb{S}_{n,t,l,r,d}$ to $\mathbb{S}_{n-1,t,l,r+j,d-1}$ where $n > 0$ and $j \ge 0$. The CAS in line \ref{cas_contractsingle} changes the state from $\mathbb{S}_{n,t,l,r}$ to $\mathbb{S}_{n,t-1,l,r+j,d-1}$ where $j \ge 0$.
\end{lemma}

\begin{proof}
We have shown in lemma \ref{l-consclean} that the CAS in line \ref{cas_clean} does not change the number of live nodes. In lemma \ref{l-comptomb} we have shown that $toCompressed$ returns a compression of the cnode $cn$ which replaces $cn$ at $in.main$ at time $t$.

Provided there is at least one single tip immediately before time $t$, the compression of the cnode $cn$ can omit at most one single tip, decreasing $r$ by one. Omitting a single tip will also decrease $d$ by one. If it is removing a single tip which is 1-way cnode, it will create a new null-inode in the trie, hence the $n + j$.

Provided there are at least $k$ null-inodes and $j$ tomb-inodes in the trie, compression may omit up to $k$ null-inodes and up to $j$ tomb-inodes. Value $d$ may decrease in the new state. If both $k$ and $j$ are $0$, then the state must have changed since a nonlive inode was detected prior to calling $clean$.

This proves the statement for CAS in line \ref{cas_clean}, the rest are either trivial or can be proved by applying a similar reasoning.
\end{proof}

\begin{lemma}\label{l-cleanstate}
If the Ctrie is in a clean state and $n$ threads are executing operations on it, then some thread will execute a successful CAS resulting in a consistency change after a finite number of execution steps.
\end{lemma}

\begin{proof}
Assume that there are $m \le n$ threads in the $clean$ operation or in the cleanup phase of the $remove$. Since the state is clean, there are no nonlive inodes, so it is trivial to show that none of these $m$ threads will invoke a CAS after their next CAS (which will be unsuccessful). This means that these $m$ threads will either complete in a finite number of steps or restart the original operation after a finite number of steps. From this point on, as shown in lemma \ref{l-betweencas}, the first CAS will be executed after a finite number of steps. Since the state is clean, there are no more nonlive inodes, so $clean$ will not be invoked. Therefore, the first CAS will result in a consistency change. Since it is the first CAS, it will also be successful.
\end{proof}

\begin{proof}[Lock-freedom]
Assume we start in some state $\mathbb{S}_{n,t,l,r,d}$. We prove there are a finite number of state changes before reaching a clean state by contradiction. Assume there is an infinite sequence of state changes. We now use results from lemma \ref{l-statechanges}. In this infinite sequence, a state change which decreases $d$ may occur only finitely many times, since no state change increases $d$. After this finitely many state changes $d = 0$ so the sequence can contain no more state changes which decrease $d$. We apply the same reasoning to $r$ -- no available state change can increase the value of $r$, so after finitely many steps $r = 0$. At this point, we can only apply state changes which decrease $n$. After finitely many state changes $n = 0$. Therefore, the assumption is wrong -- such an infinite sequence of state changes does not exist.

From corollary \ref{c-finexec} there are a finite number of execution steps between state changes, so there are a finite number of execution steps before reaching a clean state. By lemma \ref{l-cleanstate}, if the Ctrie is in a clean state, then there are an additional finite number of steps until a consistency change occurs.

This proves that some operation completes after a finite number of steps, so all Ctrie operations are lock-free.
\end{proof}

\begin{definition}
A \textbf{tip} is a cnode $cn$ such that it contains at most one reference to a tomb-inode or an snode. It may contain zero or more null-inodes, but no cnodes. If the first ancestor cnode is $k$-way where $k > 1$, then the tip has length 1.
\end{definition}

The compression operations are designed so that they collect as many null-inodes as possible, and to prevent that there are tips. We now prove that they ensure that there are no tips in the trie.

\begin{theorem}[Compactness]\label{t-space}
Assume all $remove$ operations have completed execution. Then there is at most $1$ tip of length 1 in the trie.
\end{theorem}

\begin{proof}[Compactness]
Assume that at some inode $in$ in the trie some $remove$ operation created a tip $cn$ at time $t_0$ by invoking a CAS instruction in line \ref{cas_remove}. The $remove$ operation then repeatedly tries to replace the $cn$ with a new node $mwt$ such that $mwt$ is a weak tombing of $cn$. It stops in 2 cases.

If at some time $t_1 > t_0$ the operation detects that $in$ is not a tip, it will stop. If $in$ is not a tip, then it can safely abort the compression operation, since only some other $remove$ operation performing a successful CAS in line \ref{cas_remove} at some time $t_2 > t_1$ can create a tip, and that $remove$ operation will invoke the compression again.

If at some time $t_1 > t_0$ the CAS in line \ref{cas_tomb} succeeds, then $in$ will become nonlive -- no longer a tip. Therefore, by lemma \ref{l-nonlive} $in$ does not change the value of $in.main$ $\forall t > t_0$, and all modifications to the values in that branch must occur at the first inode ancestor of $in$ -- its $parent$. Method $contractParent$ is called next in this case. If it finds that $bmp \odot flag = 0$ or $sub \neq in$, then $in$ is no longer reachable, so there are no more tips created by the current $remove$ operation at $in$ -- some other $remove$ operation may create a tip after $t_1$ at the same level and prefix as $in$, but in this case subsequent operations will be responsible for removing that tip. If $in$ is reachable, a null-inode is removed from the cnode below the $parent$ (line \ref{cas_contractnull}) or a tomb-inode is resurrected into an snode (line \ref{cas_contractsingle}). Notice that this can create a tip one level higher, but the whole procedure is repeated one level above for this reason.

The only case where we do not invoke $tombCompress$ is the root, where $parent = null$. The root can, therefore, contain at most $1$ tip of length 1.
\end{proof}

\end{document}